\documentclass[11pt,reqno]{amsart}

\usepackage[a4paper,top=3.5cm,bottom=3.5cm,left=3cm,right=3cm]{geometry}
\usepackage[T1]{fontenc} 
\usepackage[utf8]{inputenc}
\usepackage[english]{babel}

\usepackage{amssymb}
\usepackage{amsthm}
\usepackage{amsmath}
\usepackage{amsfonts}
\usepackage{amscd}

\usepackage{ bbold }
\usepackage{xcolor}
\usepackage{bbm}
\usepackage{enumerate}
\usepackage{times}
\usepackage[mathscr]{eucal}
\usepackage{indentfirst}
\usepackage{verbatim}
\usepackage{lipsum}
\usepackage{mathrsfs}
\usepackage{bm}
\usepackage{dsfont}
\usepackage{latexsym}
\usepackage{yhmath}
\usepackage{braket}

\usepackage{tikz}
\usetikzlibrary{decorations.pathreplacing,decorations.markings}
\tikzset{
  on each segment/.style={
    decorate,
    decoration={
      show path construction,
      moveto code={},
      lineto code={
        \path [#1]
        (\tikzinputsegmentfirst) -- (\tikzinputsegmentlast);
      },
      curveto code={
        \path [#1] (\tikzinputsegmentfirst)
        .. controls
        (\tikzinputsegmentsupporta) and (\tikzinputsegmentsupportb)
        ..
        (\tikzinputsegmentlast);
      },
      closepath code={
        \path [#1]
        (\tikzinputsegmentfirst) -- (\tikzinputsegmentlast);
      },
    },
  },
  mid arrow/.style={postaction={decorate,decoration={
        markings,
        mark=at position .5 with {\arrow[#1]{stealth}}
      }}},
}
\usepackage{caption}
\usepackage{float}

\allowdisplaybreaks
\makeindex

\newtheorem{theorem}{Theorem}[section]
\newtheorem{lemma}[theorem]{Lemma}
\newtheorem{proposition}[theorem]{Proposition}
\newtheorem{definition}[theorem]{Definition}

\newtheorem{remark}[theorem]{Remark}
\newtheorem{corollary}[theorem]{Corollary}
\newtheorem{assumption}[theorem]{Assumption}

\numberwithin{equation}{section}
\allowdisplaybreaks

\newcommand{\nc}{\normalcolor}

\newcommand{\E}{\mathbf{E}}
\newcommand{\R}{\mathbf{R}}
\newcommand{\C}{\mathbf{C}}
\newcommand{\Z}{\mathbf{Z}}

\newcommand{\N}{\mathbf{N}}
\newcommand{\supp}{\mathrm{supp}}

\newcommand{\ii}{\mathrm{i}}
\newcommand{\ee}{\mathrm{e}}
\newcommand{\rd}{\mathrm{d}}

\newcommand{\cO}{\mathcal{O}}

\newcommand{\cE}{\mathcal{E}}
\newcommand{\Tr}{\mathrm{Tr}}
\newcommand{\A}{\mathfrak{a}}
\newcommand{\B}{\mathfrak{p}}
\newcommand{\cau}{\phi}
\newcommand{\sing}{\mathcal{F}_{\mathrm{sing}}}
\newcommand{\tilsing}{\widetilde{\mathcal{F}}_{\mathrm{sing}}}

\newcommand{\Dela}{\frac{1}{2}\kappa_0}
\newcommand{\Delbound}{\frac{1}{6}\kappa_0}
\newcommand{\Delb}{\frac{3}{4}\kappa_0}
\newcommand{\Delc}{2\Delta}
\newcommand{\Deld}{3\Delta}

\makeatletter
\newsavebox{\@brx}
\newcommand{\llangle}[1][]{\savebox{\@brx}{\(\m@th{#1\langle}\)}%
	\mathopen{\copy\@brx\kern-0.5\wd\@brx\usebox{\@brx}}}
\newcommand{\rrangle}[1][]{\savebox{\@brx}{\(\m@th{#1\rangle}\)}%
	\mathclose{\copy\@brx\kern-0.5\wd\@brx\usebox{\@brx}}}
\makeatother

\usepackage{xcolor}

\usepackage{hyperref}
\hypersetup{
	colorlinks = true,
	citecolor = teal, 
	linkcolor=red,
	urlcolor=blue}

\title{Prethermalization for deformed Wigner matrices}

\thanks{All authors were supported by the ERC Advanced Grant ``RMTBeyond'' No.~101020331. }

\date{\today} 

\begin{document}
	\maketitle
	\vspace{0.25cm}
	
	\renewcommand{\thefootnote}{\fnsymbol{footnote}}

	\noindent
	\mbox{}%
	\hfill%
	\begin{minipage}{0.21\textwidth}
		\centering
		{L\'aszl\'o Erd\H{o}s}\footnotemark[1]\\
		\footnotesize{\textit{lerdos@ist.ac.at}}
	\end{minipage}
	\hfill%
	\begin{minipage}{0.21\textwidth}
		\centering
		{Joscha Henheik}\footnotemark[1]\\
		\footnotesize{\textit{joscha.henheik@ist.ac.at}}
	\end{minipage}
	\hfill%
	\begin{minipage}{0.21\textwidth}
		\centering
		{Jana Reker}\footnotemark[1]\\
		\footnotesize{\textit{jana.reker@ist.ac.at}}
	\end{minipage}
	\hfill%
	\begin{minipage}{0.21\textwidth}
		\centering
		{Volodymyr Riabov}\footnotemark[1]\\
		\footnotesize{\textit{vriabov@ist.ac.at}}
	\end{minipage}
	\hfill%
	\mbox{}%
	\footnotetext[1]{Institute of Science and Technology Austria, Am Campus 1, 3400 Klosterneuburg, Austria. 
	}

	\renewcommand*{\thefootnote}{\arabic{footnote}}
	\vspace{0.25cm}
	
\begin{abstract}
We prove that a class of weakly perturbed Hamiltonians of the form $H_\lambda = H_0 + \lambda W$, with $W$ being a Wigner matrix, exhibits  \emph{prethermalization}.
That is, the time evolution generated by $H_\lambda$ relaxes to its ultimate thermal state via an intermediate prethermal state with a lifetime of order $\lambda^{-2}$.
 Moreover, we obtain a general relaxation formula, expressing the perturbed dynamics via the unperturbed dynamics and the ultimate thermal state. The proof relies on a two-resolvent law for the deformed Wigner matrix $H_\lambda$. 
\end{abstract}
\vspace{0.15cm}


\vspace{0.25cm}
\normalsize
\section{Introduction} \label{sec:intro}
It is well-known (see, e.g., \cite{Gogolin}) that certain macroscopic observables in an isolated quantum system with  many interacting degrees of freedom tend to equilibrate, i.e., their expectation values become essentially constant at large times. However, if the system is coupled  to the environment~(reservoir), then  the process of relaxation to equilibrium may take different forms depending on the properties of the initial system and the structure of the perturbation. 

In this work, we consider a weakly coupled system of the form
\begin{equation} \label{eq:Hamiltonian}
	H_\lambda := H_0  + \lambda W\,,
\end{equation}
where $H_0$ is a single-body or a many-body \emph{Hamiltonian},  $W$ is an energy preserving (Hermitian) \emph{perturbation}, and $\lambda$ is a small coupling constant. For our phenomenological study, we consider a mean-field random perturbation that  couples all modes. Following the extensive physics literature \cite{Wigner, Deutsch1991, NationPorras, DRpretherm, DRrelax1, DRrelax2}, we choose the perturbation $W$ to be a \emph{Wigner random matrix}, i.e., a random matrix with centered, independent identically distributed (i.i.d.) entries (modulo the Hermitian symmetry). 

The central object of our study is the perturbed time evolution of the quantum expectation value
\begin{equation} \label{eq:tracing}
	\langle A \rangle_{P_\lambda(t)} := \Tr [P_\lambda(t) A]
\end{equation}
of an \emph{observable} $A$, compared to the unperturbed evolution $\langle A \rangle_{P_0(t)} = \Tr [P_0(t) A]$, which is considered known. Here
\begin{equation}\label{eq:timeev}
	P_\lambda(t):=\ee^{-\ii t H_\lambda}P\ee^{\ii t H_\lambda} \qquad \text{resp.} \qquad 	P_0(t):=\ee^{-\ii t H_0}P\ee^{\ii t H_0}
\end{equation}
denote the Heisenberg time evolution of an initial \emph{state} $P$ governed by the (un)perturbed Hamiltonian.  
We point out that the \emph{unperturbed} evolution strongly depends on all its constituents 
and hence, it might exhibit qualitatively different and generally complex behavior.

In contrast, the \emph{perturbed} system relaxes to equilibrium via a robust mechanism, and it can be described by a fairly simple general \emph{relaxation formula}
\begin{equation}\label{eq:keyeq}
	\langle A\rangle_{P_\lambda(t)}\approx\langle A\rangle_{\widetilde{P}_\lambda}+|g_\lambda(t)|^2
	\big[\langle A\rangle_{P_0(t)}-\langle A\rangle_{\widetilde{P}_\lambda}\big], \qquad  g_\lambda(t): = \ee^{-\alpha\lambda^2t},
	\quad \alpha>0,
\end{equation}
where $\widetilde{P}_\lambda$ is the thermal state of the composite system~\eqref{eq:Hamiltonian}. In this form, Eq.~\eqref{eq:keyeq} is first mentioned in~\cite[Eq.~(40)]{NationPorras}, where it describes the time dependence of the expectation of an observable in a nonintegrable system after perturbation by a random matrix.

The relaxation formula \eqref{eq:keyeq} shows convergence to the thermal state at an exponential rate on time scales of 
order $\lambda^{-2}$ (in agreement with \emph{Fermi's golden rule}),
 but it also carries more refined information about the role of the unperturbed dynamics
in the process. A particularly interesting case occurs if both the perturbed and unperturbed systems equilibrate but do not approach the same limiting value. This often happens if $H_0$ has an additional symmetry (conserved quantity) that is broken by the perturbation. If the time scale $\lambda^{-2}$ of the perturbed equilibration is smaller than that of the unperturbed one, then the former robust process eclipses the latter. In particular, the precise form of $\langle A\rangle_{P_0(t)}$ in~\eqref{eq:keyeq} is irrelevant whenever the prefactor $|g_\lambda(t)|^2$ is already exponentially small. In the opposite case, however, the equilibration of the perturbed dynamics happens in two stages. This phenomenon, known as \emph{prethermalization} in the physics literature, was first described in a paper by Moeckel and Kehrein~\cite{MK2008}. We remark, however, that this terminology was already used to describe a different phenomenon a few years earlier~\cite{BBW2004}. 
  
Nowadays, prethermalization has been extensively studied both experimentally (see, e.g., the review~\cite{LGS2016}) and theoretically (e.g., in~\cite{BISZ2016,DieplingerBera,GalloneLangella,KWE2011,Reimann2016,Ueda2020, MRdR2019}, see also the review~\cite{MIKUreview}). Reimann and Dabelow \cite{DRpretherm} studied the first relaxation stage of a prethermalization process, which is governed by $H_0$. More precisely, assuming that $P_0(t)$ relaxes to a non-thermal steady state, they find that the perturbed time evolution $\langle A\rangle_{P_\lambda(t)}$ with a sufficiently weak perturbation ($\lambda \ll 1$) closely follows the unperturbed time evolution $\langle A\rangle_{P_0(t)}$ for times $t \ll \lambda^{-2}$. In particular, the perturbed time evolution $\langle A\rangle_{P_\lambda(t)}$ is close to the non-thermal steady state of $H_0$ for times $1 \ll t \ll \lambda^{-2}$.\footnote{Unperturbed systems $H_0$, for which the time evolution $\langle A \rangle_{P_0(t)}$ does not approach the microcanonical prediction of equilibrium statistical mechanics, but nevertheless have a large-$t$ limit (a non-thermal steady state), are studied in \cite{BalzReimann}. In a sense, these systems exhibit prethermalization, although they do not approach thermal equilibrium after the first steady state is reached.} 
The authors of \cite{DRpretherm} further extended their principal approach to a general study of relaxation theory for perturbed quantum dynamics in~\cite{DRrelax1, DRrelax2}. These works now include all times and also the \emph{strong} coupling regime (in case of banded matrices), which yields a characteristic power-law time decay (given more precisely by a Bessel function) instead of the exponential decay in~\eqref{eq:keyeq}. The theoretical model is then applied to several examples and compare the prediction to numerical and experimental works (see also Dabelow's PhD thesis~\cite{DabelowThesis} for further details). Finally, we also mention that prethermalization in the form of the existence of an effectively conserved quantity for very long times has been rigorously established in \cite{AdRHH2016} for periodically driven quantum systems if the frequency is large compared with the size of the driving potential.

In this paper, we approach prethermalization from the viewpoint of random matrix theory, interpreting the unperturbed Hamiltonian $H_0 \equiv H_0(N) \in \C^{N \times N}$ as a fixed sequence of bounded self-adjoint deterministic matrices and the perturbation $W\equiv W(N)$ as an $N\times N$ Wigner random matrix. Our Hamiltonian $H_\lambda$ in the setting of~\eqref{eq:Hamiltonian} is also called~\emph{deformed Wigner matrix} in random matrix theory, or it can be viewed as a Wigner random matrix with nonzero expectation.
Wigner matrices are encountered in many related physics models, e.g., the recent rigorous study of thermalization problems~\cite{ETHpaper,thermalization,multiG}. Here, the key technical result is a strong concentration property of the resolvent  $G(z)= (H_\lambda-z)^{-1}$
or products of several resolvents around their deterministic approximation. Such results are commonly called \emph{multi-resolvent
global or local laws}, depending on the distance of the spectral parameter from the spectrum.
For example, a typical \emph{two-resolvent} law computes 
\begin{equation}\label{eq:GAGA}
\llangle G(z_1)A_1G(z_2)A_2\rrangle=\llangle(H_\lambda-z_1)^{-1}A_1(H_\lambda-z_2)^{-1}A_2\rrangle
\end{equation}
to leading order in $N$, \nc
where $\llangle\cdot\rrangle$ denotes the normalized trace, $z_1,z_2\in\C\setminus\R$, and $A_1,A_2\in\C^{N\times N}$ are deterministic matrices. Using functional calculus, the resolvents can be replaced by more 
general and even $N$-dependent functions, thus linking~\eqref{eq:GAGA} to the Heisenberg time evolution.
Recent work~\cite{equipart} establishes a multi-resolvent local law for deformed Wigner matrices in the bulk regime of the spectrum, which motivated our study of perturbed quantum systems.

\subsection{Description of the main results}
The principal goal of this work is a rigorous proof of the relaxation formula (Corollary \ref{cor:relax}) and prethermalization (Corollary \ref{cor:PreT1}) for perturbed quantum Hamiltonians of the form \eqref{eq:Hamiltonian}. We thereby assume that the unperturbed Hamiltonian $H_0$ has a (locally) regular limiting density of states $\rho_0$ around a reference energy $E_0$ and only energies in a microscopically large but macroscopically small interval $I_\Delta:= [E_0 - \Delta, E_0+ \Delta]$ are populated by the initial state $P$ (similar assumptions are made in \cite{DRpretherm, DRrelax1, DRrelax2}). We then show the following corollaries of our main Theorem \ref{thm:main}: 
\begin{itemize}
\item[Cor.~\ref{cor:relax}:] The relaxation formula \eqref{eq:keyeq} holds generally for short and long kinetic times, i.e.~$t \ll \lambda^{-2}$ and $t \gg \lambda^{-2}$, corresponding to $|g_\lambda(t)|^2 \approx 1$ and $|g_\lambda(t)|^2 \approx 0$, respectively. At intermediate times, $t \sim \lambda^{-2}$ it is generally \emph{not} valid, unless the quadratic forms $\langle \bm u_j, A \bm u_j \rangle$ of overlaps with the eigenvectors $\bm u_j$ of $H_0$ behave regularly in $j$ (cf.~Definition \ref{def:LOR}). This happens, e.g., if $H_0$ satisfies the Eigenstate Thermalization Hypothesis (ETH). 
\item[Cor.~\ref{cor:PreT1}:] Assuming that the unperturbed time evolution has a long time limit $\langle A \rangle_{P_0(t)} \overset{t \to \infty}{\longrightarrow} \langle A \rangle_{P_{\mathrm{pre}}}$, such that the \emph{prethermal state} $P_{\rm pre}$ is distinguishable from the thermal state $\widetilde{P}_{\lambda}$ of the perturbed system, i.e.~$\langle A \rangle_{P_{\mathrm{pre}}} \neq \langle A \rangle_{\widetilde{P}_\lambda}$
 (cf.~Definition \ref{def:preTcond}), we show the characteristic two-step relaxation of a prethermalization process; see Figure \ref{fig-2plateau}. 
\end{itemize}

\begin{figure}[h]
	\begin{center}
		\begin{tikzpicture}[scale=\textwidth/15.2cm]
			\nc
			\draw[->] (-0.5,0) -- (10,0);
			\draw (10,0) node[below=1pt] {$t$};
			\draw[->] (0, -0.5) -- (0,5);
			\draw (0,5) node[left=1pt] {$\langle A\rangle_{P_\lambda(t)}$};

			\draw[dotted] (0,3.5) -- (10,3.5);
			\draw (0,3.5) node[left=1pt] {$\langle A\rangle_{P_{\rm pre}}$};
			\draw[dotted] (0,1.5) -- (10,1.5);
			\draw (0,1.5) node[left=1pt] {$\langle A\rangle_{\widetilde{P}_\lambda}$};
			
			\draw (5,0) node[below=1pt] {$\sim \lambda^{-2}$};
			
			\draw[red,thick,dashed] (0.25,4.5) .. controls (0.75,4.5) .. (1,4.142);
			\draw[red,thick,smooth] (1,4.142) .. controls (1.5,3.4889) .. (3,3.4889);
			
			\draw[red,thick,smooth] plot[variable=\x,domain=3:10] (\x,{-2/(1+exp((-(\x-5))/0.4))+3.5});
			
		\end{tikzpicture}
	\end{center}
	\captionof{figure}{Depicted is a schematic graph of the prethermalization process: For times $1 \ll t\ll \lambda^{-2}$, the perturbed time evolution of the quantum expectation $\langle A \rangle_{P_\lambda(t)}$ (see \eqref{eq:tracing}) is close to the quantum expectation of $A$ in the prethermal state $P_{\rm pre}$, i.e.~$\langle A \rangle_{P_\lambda(t)} \approx \langle A \rangle_{P_{\mathrm{pre}}}$. For times $t \gg \lambda^{-2}$, we have that $\langle A \rangle_{P_\lambda(t)}$ is close to the limiting thermal quantum expectation, i.e.~$\langle A \rangle_{P_\lambda(t)} \approx \langle A \rangle_{\widetilde{P}_{\lambda}}$. This value is typically different from the prethermal quantum expectation $\langle A \rangle_{P_{\mathrm{pre}}}$.
		This means, the ultimate relaxation of $\langle A \rangle_{P_\lambda(t)}$ towards  $\langle A \rangle_{\widetilde{P}_{\lambda}}$ happens via an intermediate \emph{prethermal} value $\langle A \rangle_{P_{\mathrm{pre}}}$ in two steps, whose time scales are separated by~$\lambda^{-2}$. }\label{fig-2plateau}
\end{figure}
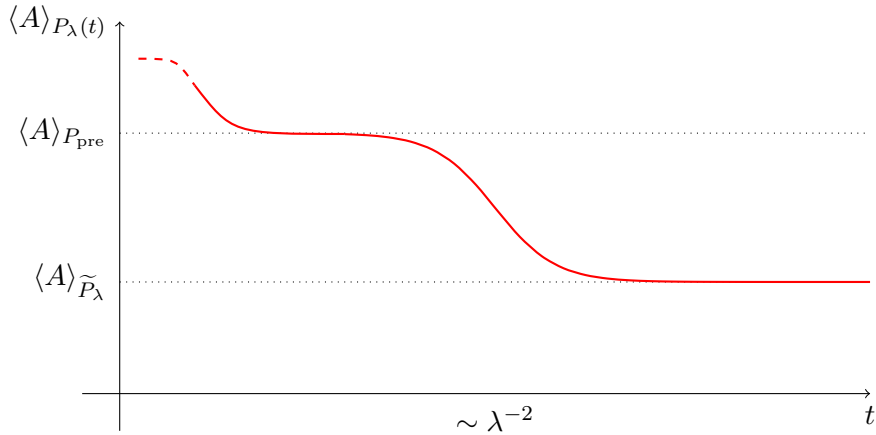

\subsection{Outline of the paper}
The general task in this paper is to approximately evaluate the random time evolution $\langle A \rangle_{P_\lambda(t)}$ from \eqref{eq:tracing}. This is carried out in several steps 
summarized schematically in Figure \ref{fig-structure}.
First, in Theorem~\ref{thm:main}~(a) in Section \ref{subsec:relaxation}, 
the true  time evolution $P_\lambda(t)$ is expressed 
as a linear combination of the unperturbed time evolution $P_0(t)$ 
and another deterministic time-dependent object $\widetilde{P}_{\lambda, t}$ that is conceptually simpler
than $P_\lambda(t)$.
Then, in Theorem~\ref{thm:main}~(b), we identify a time-independent state $\widetilde{P}_{\lambda}$
as the large time limit of
$\widetilde{P}_{\lambda, t}$. Combining both parts of Theorem \ref{thm:main}, we arrive at Corollary \ref{cor:relax}, which establishes the relaxation formula \eqref{eq:keyeq} at small and large kinetic times. As mentioned above, at intermediate times, it holds only for observables having the \emph{local overlap regularity (LOR) property} (see Definition \ref{def:LOR}).
In the subsequent Section~\ref{subsec:pretherm}, dropping the LOR property, but assuming additionally
that the unperturbed Hamiltonian $H_0$ and the initial state $P$ have the 
\emph{prethermalization property} (see Definition \ref{def:preTcond}), we obtain the characteristic two-scale relaxation of $P_\lambda(t)$ towards $\widetilde{P}_{\lambda}$ via an intermediate \emph{prethermal state} $P_{\rm pre}$ (see Corollary~\ref{cor:PreT1} and Figure \ref{fig-2plateau}). 
As an additional result, in Theorem \ref{thm:PreT2} in Section \ref{subsec:mcensemble}
we relate $\widetilde{P}_\lambda$ to the microcanonical ensemble of $H_\lambda$, called $P_\lambda^{\rm (mc)}$, which is  independent of the initial state $P$. Finally, our results are illustrated by two simple examples in Section \ref{subsec:example}. 

While most proofs are given in Section \ref{sec:proofs}, some auxiliary results and additional proofs are deferred to Appendix \ref{app:app}.

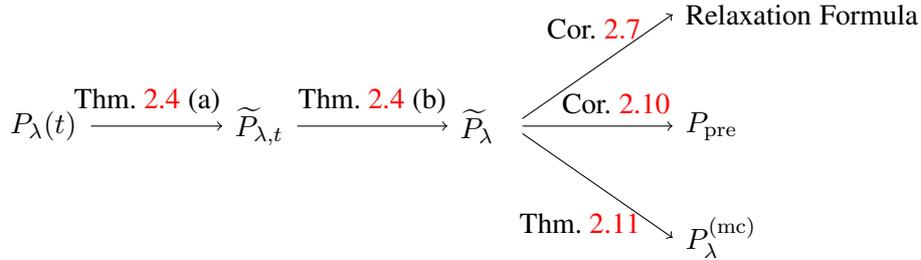
\begin{figure}[h]
	\begin{center}
		\begin{tikzpicture}[scale=\textwidth/15.2cm]
			\nc
			\draw (0,0) node[right=1pt] {$P_\lambda(t)$};
			\draw (3,0) node[right=1pt] {$\widetilde{P}_{\lambda,t}$};
			\draw (6,0) node[right=1pt] {$\widetilde{P}_\lambda$};
			\draw (9,1.5) node[right=1pt] {Relaxation Formula};
			\draw (9,0) node[right=1pt] {$P_{\rm pre}$};
			\draw (9,-1.5) node[right=1pt] {$P_\lambda^{({\rm mc})}$};
			
			\draw[->] (1.25,0) -- (3,0);
			\draw[->] (4,0) -- (6,0);
			\draw[->] (7,0.1) -- (9,1.5);
			\draw[->] (7,0) -- (9,0);
			\draw[->] (7,-0.1) -- (9,-1.5);
			
			\draw (2,0) node[above=1pt] {Thm.~\ref{thm:main} (a)};
			\draw (5,0) node[above=1pt] {Thm.~\ref{thm:main} (b)};
			\draw (8.75,1) node[above left=1pt] {Cor.~\ref{cor:relax}};
			\draw (8.25,0) node[above=1pt] {Cor.~\ref{cor:PreT1}};
			\draw (8.75,-1) node[below left=1pt] {Thm.~\ref{thm:PreT2}};
			
		\end{tikzpicture}
	\end{center}
	\captionof{figure}{The structure of our main results.}\label{fig-structure}
\end{figure}

\subsection{Notation}
For positive quantities $f,g$ we write $f\lesssim g$ (or $f=\mathcal{O}(g)$) and $f\sim g$ if $f \le C g$ or $c g\le f\le Cg$, respectively, for some constants $c,C>0$ which only depend on the constants appearing in the moment condition (see \eqref{eq:moment_bound}) and the definition of the set of admissible energies (see \eqref{eq:admiss_spec}).
In informal explanations, we frequently use the notation $f \ll g$, which indicates that $f$ is "much smaller" than $g$. Moreover, we shall also write $w \approx z$ to indicate the closeness of $w, z \in \C$ with a not precisely specified error. 

For any natural number $n$ we set $[n]: =\{ 1, 2,\ldots ,n\}$. Matrix entries are indexed by lowercase Roman letters $a, b, c, ...$ and $ i,j,k, ...$ from the beginning or the middle of the alphabet and unrestricted sums over $a,b,c,...$ and $ i,j,k,...$ are always understood to be over $[N] = \{1,...,N\}$. 

We denote vectors by bold-faced lowercase Roman letters ${\bm x}, {\bm y}\in\C ^N$, for some $N\in\N$. Vector and matrix norms, $\lVert {\bm x}\rVert$ and $\lVert A\rVert$, indicate the usual Euclidean norm and the corresponding induced matrix norm. For any $N\times N$ matrices $A,B$ we use the notations $\llangle A\rrangle:= N^{-1}\mathrm{Tr}  A$ to denote the normalized trace of $A$ and $\langle A \rangle_B := \Tr[AB]$ is the trace of the product $AB$. We denote the spectrum of a matrix or operator $A$ by $\sigma(A)$. 
Moreover, for vectors ${\bm x}, {\bm y}\in\C^N$ and matrices  $A\in\C^{N\times N}$ we define 
\[
\langle {\bm x},{\bm y}\rangle:= \sum_{i} \overline{x}_i y_i\,, \qquad A_{\bm x \bm y} := \langle \bm x, A \bm y \rangle \,. 
\]
For a unit vector $\bm v \in \C^{N}$ we shall also use the notation $\ket{\bm v} \bra{\bm v}$ for the projection onto the one-dimensional subspace spanned by $\bm v$.

Finally, we use the concept of ``with very high probability'' \emph{(w.v.h.p.)} meaning that for any fixed $C>0$, the probability of an $N$-dependent event is bigger than $1-N^{-C}$ for $N\ge N_0(C)$. We introduce the notion of \emph{stochastic domination} (see e.g.~\cite{semicirclegeneral}): given two families of non-negative random variables
\[
X=\left(X^{(N)}(u) : N\in\N, u\in U^{(N)} \right) \quad \mathrm{and}\quad Y=\left(Y^{(N)}(u) : N\in\N, u\in U^{(N)} \right)
\] 
indexed by $N$ (and possibly some parameter $u$  in some parameter space $U^{(N)}$), 
we say that $X$ is stochastically dominated by $Y$, if for all $\xi, C>0$ we have 
\begin{equation}
	\label{stochdom}
	\sup_{u\in U^{(N)}} \mathbf{P}\left[X^{(N)}(u)>N^\xi  Y^{(N)}(u)\right]\le N^{-C}
\end{equation}
for large enough $N\ge N_0(\xi,C)$. In this case we use the notation $X\prec Y$ or $X= \mathcal{O}_\prec(Y)$.

\subsection*{Acknowledgments} We thank Peter Reimann and Lennart Dabelow for helpful comments.

\section{Main Results} \label{sec:main}

In Section \ref{subsec:Assumptions}, we give the precise definition of Wigner random matrices and the assumptions on the Hamiltonian, observables and states under consideration in \eqref{eq:Hamiltonian}--\eqref{eq:timeev}. Afterward, we formulate our main results in Sections \ref{subsec:relaxation}--\ref{subsec:mcensemble}. Finally, in Section \ref{subsec:example} we discuss our findings in the context of two simple examples.

\subsection{Assumptions} \label{subsec:Assumptions} 
We begin with formulating the assumption on the Wigner matrix $W$. 
		\begin{assumption}[Wigner matrix]  \label{ass:W}
		Let $W \equiv W(N) = (w_{ij})_{i,j \in [N]}$ from \eqref{eq:Hamiltonian} be a real symmetric or complex Hermitian random matrix $W = W^*$ with independent entries distributed according to the laws $w_{ij} \stackrel{\mathrm{d}}{=} N^{-1/2}\chi_{\mathrm{od}}$ for $i < j$ and $w_{jj} \stackrel{\mathrm{d}}{=} N^{-1/2}\chi_{\mathrm{d}}$. The random variables $\chi_{\mathrm{od}}$ and $\chi_{\mathrm{d}}$ satisfy the following assumptions:\footnote{A careful examination of our proof reveals that the entries of $W$ need not be distributed identically. Indeed, only the matching of the second moments is necessary, but higher moments can differ.}
		We assume that $\chi_{\mathrm{d}}$ is a centered real random variable, and $\chi_{\mathrm{od}}$ is a real or complex random variable with $\E \chi_{\mathrm{od}} = 0$ and $\E |\chi_{\mathrm{od}}|^2 = 1$.
		
		Furthermore, we assume the existence of higher moments, namely
		\begin{equation} \label{eq:moment_bound}
			\E |\chi_{\mathrm{d}}|^p + \E |\chi_{\mathrm{od}}|^p \le C_p,
		\end{equation}
		for all $p\in \N$, where $C_p$ are positive constants.
	\end{assumption}
For concreteness, we focus on the complex case with the additional assumptions $\E \chi_{\mathrm{od}}^2 = 0$ and $\E |\chi_{\mathrm{d}}|^2 = 1$; all other cases can also be handled as in \cite{ETHpaper}. The precise conditions on the Wigner matrix only play a role in the underlying \emph{two resolvent global law} (Proposition \ref{prop:LL}).

For the Hamiltonian $H_0 \equiv H_0(N)$ in \eqref{eq:Hamiltonian} we assume the following. 
	\begin{assumption}[$H_0$ and its density of states] \label{ass:H0}
		The Hamiltonian $H_0 \equiv H_0(N)$ is deterministic, self-adjoint $H_0 = H_0^*$, and uniformly bounded $\Vert H_0 \Vert \lesssim 1$. We denote the resolvent of $H_0$ at any spectral parameter $z \in \C\setminus \R$ by
		\begin{equation*}
M_0(z):= \frac{1}{H_0 -z}\,.
		\end{equation*}
		Moreover, we assume the following: 
		\begin{itemize}
\item[(i)] 	There exists a compactly supported measurable function $\rho_0 : \R \to [0,+\infty)$ with 
$$\int_\R \rho_0(x) \rd x = 1$$
 and two positive sequences $\epsilon_0(N)$ and $\eta_0(N)$, both converging to zero as $N\to\infty$,
  such that, uniformly in $z \in \C\backslash\R$ with $\eta:=|\Im z| \ge \eta_0 \equiv \eta_0(N)$, we have
\begin{equation} \label{eq:rho0}
	\llangle[\bigl] M_0(z)\rrangle[\bigr] = m_0(z) + \mathcal{O}(\epsilon_0) \quad \text{with} \quad \epsilon_0 \equiv \epsilon_0(N)\,.
\end{equation}
Here, 
\begin{equation} \label{eq:m0}
m_0(z):= \int_\R \frac{\rho_0(x)}{x-z}\mathrm{d}x
\end{equation} 
is the Stieltjes transform of $\rho_0$. We refer to $\rho_0$ as the \emph{limiting density of states}, and to $\supp(\rho_0)$ as the \emph{limiting spectrum} of $H_0$. 
\item[(ii)] For small positive constants $\kappa,c>0$, we define the set of \emph{admissible energies} $\sigma_{\mathrm{adm}}^{(\kappa,c)}$ in the limiting spectrum of $H_0$ by\footnote{Here, $C^{1,1}(J)$ denotes the set of continuously differentiable functions with a Lipschitz-continuous derivative on an interval $J$, equipped with the norm $\lVert f\rVert_{C^{1,1}(J)} := \lVert f\rVert_{C^1(J)} + \sup\limits_{x,y\in J: x\neq y} \frac{|f'(x)-f'(y)|}{|x-y|}$.}
\begin{equation} \label{eq:admiss_spec}
\sigma_{\mathrm{adm}}^{(\kappa,c)} := \left\{x \in \supp(\rho_0) : \inf_{|y-x|\le\kappa}\rho_0(y) > c,\, \lVert\rho_0\rVert_{C^{1,1}([x-\kappa, x+\kappa])} \le 1/c\right\}.
\end{equation}
We assume that for some positive $\kappa,c> 0$, $\sigma_{\mathrm{adm}}^{(\kappa,c)}$ is not empty.
		\end{itemize}
	\end{assumption}

Assuming that the set of admissible energies in~\eqref{eq:admiss_spec} is non-empty ensures that there is a part of the limiting spectrum $\mathrm{supp}(\rho_0)$, where the limiting density of states $\rho_0$ behaves regularly, i.e.~it is strictly positive and sufficiently smooth. Finally, we formulate our assumptions on the observables and states considered in \eqref{eq:tracing} and~\eqref{eq:timeev}. 
\begin{assumption}[States and observables] \label{ass:stateandobs}
Given Assumption \ref{ass:H0}, we first pick a \emph{reference energy}
\begin{equation} \label{eq:E0}
E_0 \in \sigma_{\mathrm{adm}}^{(\kappa_0,c_0)} \quad  \text{for some}  \quad \kappa_0, c_0>0,
\end{equation}
and further introduce $I_\delta:=[E_0-\delta,E_0+\delta]$  for any $0<\delta<\kappa_0$. 
Moreover, take an \emph{energy width} $\Delta \in (0, \Delbound)$ and let
 $\Pi_\Delta := \mathds{1}_{I_\Delta}(H_0)$ be the spectral projection of 
 $H_0$ onto the interval $I_\Delta$. 
 Then, we assume the following: 
\begin{itemize}
\item[(i)] The (deterministic) initial state $P \equiv P(N) \in \C^{N \times N}$ in \eqref{eq:timeev} is a state in the usual sense~($P = P^*$, $0 \le P \le 1$, and $\Tr [P] = 1$), and is localized in $I_\Delta$, i.e.
\begin{equation} \label{eq:Plocalize}
	P = \Pi_\Delta P \Pi_\Delta\,. 
\end{equation}
\item[(ii)] The observable $A \equiv A(N) \in \C^{N \times N}$ is deterministic, self-adjoint and satisfies $\Vert A \Vert \lesssim 1$. 
\end{itemize}
\end{assumption}
Note that we assume only the state $P$ to be localized in $I_\Delta$ and not the observable. However, 
by inspecting the proof, we see that $A$ and $P$ play essentially symmetric roles, and thus, our results hold verbatim if we assume localization of $A$ instead of $P$. Moreover, for ease of notation, we drop the $N$-dependence of all the involved matrices. 

\subsection{Relaxation of perturbed quantum dynamics} \label{subsec:relaxation} In this section, we present our main result on the time evolution of the random quantum expectation $\langle A \rangle_{P_\lambda(t)}$ from \eqref{eq:tracing}.
Its relaxation is described in two steps, hence Parts (a) and (b) in the following theorem. In the first 
step, we eliminate the randomness and  identify the leading deterministic part of $\langle A \rangle_{P_\lambda(t)}$
in terms of $M_0$, the unperturbed resolvent. In the second step, we consider the short and long-time limits
of the leading term. Further explanatory comments come after the theorem and 
in Remark \ref{rmk:main} below. 

\begin{theorem}[Relaxation of perturbed dynamics]\label{thm:main}
Let $H_\lambda = H_0 + \lambda W$ be a perturbed Hamiltonian like in \eqref{eq:Hamiltonian} with $\lambda > 0$, whose constituents satisfy Assumptions \ref{ass:W} and \ref{ass:H0}, respectively. Pick a reference energy $E_0$ like in \eqref{eq:E0}. Let $P$ be a state satisfying Assumption \ref{ass:stateandobs}~(i) for some energy width $\Delta > 0$ and $A$ an observable satisfying Assumption \ref{ass:stateandobs}~(ii). 

Then, using the notations \eqref{eq:tracing} and~\eqref{eq:timeev}, we have the following two approximation statements:
\begin{itemize}
\item[(a)]  \emph{[Relaxation in the kinetic limit]} The perturbed dynamics $\langle A \rangle_{P_\lambda(t)}$ satisfies
\begin{equation} \label{eq:main}
	\langle A \rangle_{P_\lambda(t)}=|g_\lambda(t)|^2 \langle A\rangle_{P_0(t)}+ \langle A \rangle_{\widetilde{P}_{\lambda,t}}+ \cE\,,
\end{equation}
 where we denoted 
\begin{equation} \label{eq:gfunction}
	g_\lambda(t) := \ee^{-\alpha \lambda^2 t} \quad \text{with} \quad \alpha := \pi \rho_0(E_0) \,. 
\end{equation}
Moreover, we introduced\footnote{Recall that the imaginary part of a matrix $B\in\C^{N\times N}$ is given by $\Im B=\frac{1}{2i}(B-B^*)$.}
\begin{equation} \label{eq:kernelstate}
	\widetilde{P}_{\lambda,t} := \frac{\int_\R \int_\R \Im M_0(x + \ii \alpha \lambda^2) \, K_{\lambda, t}(x-y) \, \langle \Im M_0(y + \ii \alpha \lambda^2)  \rangle_P \, \rd x \rd y}{\int_\R \Tr [\Im M_0(x + \ii \alpha \lambda^2)] \,  \langle \Im M_0(x + \ii \alpha \lambda^2)  \rangle_P \,  \rd x }\,, 
\end{equation}
with an explicit kernel given by
\begin{equation} \label{eq:kernel}
K_{\lambda, t}(u) := \frac{1}{\pi}\frac{2 \alpha \lambda^2 }{u^2 + (2 \alpha \lambda^2)^2} \left( 2\alpha \lambda^2 t \frac{\sin(tu)}{tu} - \cos(tu) + \ee^{-2\alpha \lambda^2 t}\right) \,. 
\end{equation}
Finally, we have $\cE = \mathcal{O}(\cE_0) + \mathcal{O}_\prec (C(\lambda, t)/\sqrt{N})$ for some constant $C(\lambda, t) > 0$ and for every fixed $T \in (0, \infty)$, the deterministic error $\cE_0=\cE_0(\lambda, t, \Delta, N)$ satisfies
\begin{equation} \label{eq:triplelimvanish}
	\lim_{\Delta\rightarrow0}\lim_{\substack{t\rightarrow\infty\\\lambda\rightarrow0\\ \lambda^2t=T}}\lim_{N\rightarrow\infty} \cE_0=0\,.
\end{equation}
\item[(b)] \emph{[Long and short kinetic time limit]}
Defining
\begin{equation} \label{eq:kernelstateapprox}
	\widetilde{P}_{\lambda} := \frac{\int_\R  \Im M_0(x + \ii \alpha \lambda^2) \,  \langle \Im M_0(x + \ii \alpha \lambda^2) \rangle_P  \, \rd x }{\int_\R \Tr [\Im M_0(x + \ii \alpha \lambda^2)] \,  \langle \Im M_0(x + \ii \alpha \lambda^2)  \rangle_P \,  \rd x }\,, 
\end{equation}
it holds that
\begin{equation} \label{eq:mainb}
\langle A \rangle_{\widetilde{P}_{\lambda,t}} = \big(1 - |g_\lambda(t)|^2\big) \langle A \rangle_{\widetilde{P}_{\lambda}} + \mathcal{R}\,, 
\end{equation}
where, for every fixed $T \in (0,\infty)$, the error term $\mathcal{R} = \mathcal{R}(\lambda, t, \Delta, N)$ satisfies
\begin{equation} \label{eq:convergencetildeP}
\begin{split}
\limsup_{\Delta\rightarrow0} \limsup_{\substack{t\rightarrow\infty\\\lambda\rightarrow0\\ \lambda^2 t=T}} \limsup_{N\rightarrow\infty} |\mathcal{R}| \lesssim T\ee^{-2\alpha T}.
\end{split}
\end{equation}
\end{itemize}
\end{theorem}
We point out that the error $\cE$ in~\eqref{eq:main} naturally consists of two parts, a deterministic and a stochastic one. The stochastic part of order $\mathcal{O}_\prec (C(\lambda, t)/\sqrt{N})$ is obtained from a \emph{global law} for two resolvents of the random matrix $H_\lambda$ (see \eqref{eq:globallambda} below); the deterministic part $\mathcal{O}(\cE_0)$ is obtained from estimating the deterministic leading term in~\eqref{eq:terms}. 

Note that the error $\mathcal{R}$ is small compared to the first  term
 in the rhs. of~\eqref{eq:mainb} only in the regime where $T$ is large,
 in particular $\langle A \rangle_{\widetilde{P}_{\lambda,t}} $ 
 converges to $\langle A \rangle_{\widetilde{P}_{\lambda}}$ exponentially fast.
In the small $T$ regime, both terms on the right-hand side of~\eqref{eq:mainb} vanish  linearly in $T$. 
We chose the above formulation~\eqref{eq:mainb} because, in this way, it
relates directly to the relaxation formula \eqref{eq:keyeq} (see Corollary \ref{cor:relax} below).

\begin{remark} \label{rmk:main} We have two further comments on Theorem \ref{thm:main}. 
	\begin{itemize}
		\item[(i)] The triple limits in~\eqref{eq:triplelimvanish} and \eqref{eq:convergencetildeP} consist of a thermodynamic limit ($N\rightarrow\infty$), a kinetic limit or \emph{van Hove limit} ($t\rightarrow\infty$ and $\lambda\rightarrow0$ while $\lambda^2t$ is fixed), and an infinitesimal spectral localization ($\Delta\rightarrow0$). Note that the \emph{kinetic time parameter} $T=\lambda^2t$ is natural, as the time scale prescribed for relaxation in the physics literature, e.g., by explicit analysis of the Pauli master equation in~\cite[Sect.~5.2.6]{MIKUreview}, is $\cO(\lambda^{-2})$. The $\Delta\rightarrow0$ limit is needed only to ensure that the mean level spacing is approximately constant near $E_0$ on the scale~$\Delta$. If the density of states is flat, the $\Delta\rightarrow0$ limit can be omitted. We emphasize that the error terms in Theorem~\ref{thm:main} are explicit in the sense that their dependence on the scaling parameters $N$, $t$, $\lambda$, and $\Delta$ is tracked throughout the proof. The limit in~\eqref{eq:triplelimvanish} is then 
		the natural order of limits in which these errors vanish. 
		Finally, we remark that the explicitly tracked errors allow for certain combined limits, although for simplicity, we do not pursue these extensions.
		\item[(ii)] The idea behind \eqref{eq:kernelstateapprox}--\eqref{eq:convergencetildeP} is that the kernel \eqref{eq:kernel} is an approximate delta function with $T=\lambda^2t$-dependent magnitude. More precisely, its Fourier transform\footnote{\label{ftn:FT}We use the convention that the Fourier transform of $f \in L^1(\R)$ is defined as
			$\widehat{f}(p):= (2 \pi)^{-1/2} \int_{\R} f(x) \ee^{-\ii px} \rd x$. }
		\begin{equation*}
			\widehat{K}_{\lambda, t}(p) = \frac{1}{\sqrt{2 \pi}}  \, \big(1 - \ee^{-2 \alpha \lambda^2(t - |p|)}\big) \, \mathds{1}(|p|\le t)
		\end{equation*}
		converges uniformly in compact intervals to a constant. That is, for every fixed $T \in (0,\infty)$ and for every compact set $\Omega \subset \R$, we have
		\begin{equation*}
			\lim_{\substack{t\rightarrow\infty\\\lambda\rightarrow0\\ \lambda^2 t=T}} \sup_{p \in \Omega}\left| \widehat{K}_{\lambda, t}(p)  - \frac{1}{\sqrt{2 \pi} } \big(1 - \ee^{-2\alpha T}\big)\right| = 0\,. 
		\end{equation*}
	However, since $x \mapsto \Im M_0(x + \ii \alpha \lambda^2)$ is only regular on scale $\lambda^2$, the approximation $$K_{\lambda, t}(x-y) \approx \big(1 - \ee^{-2\alpha \lambda^2 t}\big) \delta(x-y),$$ 
	used in heuristically obtaining \eqref{eq:mainb} from \eqref{eq:kernelstate} and \eqref{eq:kernelstateapprox}, is \emph{not} generally valid unless~$T$ is very small or very large, as \eqref{eq:convergencetildeP} indicates.
	\end{itemize}
\end{remark}

In the remaining part of the current Section \ref{subsec:relaxation}, we connect Theorem \ref{thm:main} to the relaxation formula \eqref{eq:keyeq}.  As a preparation, we formulate the following \emph{local overlap regularity property}, required in Corollary \ref{cor:relax}~(c) below. For this purpose, let  $\mu_j$ and $\bm u_j$ denote the eigenvalues and corresponding normalized eigenvectors of
\begin{equation} \label{eq:specdecH0mainres}
H_0 = \sum_j \mu_j \ket{\bm u_j} \bra{\bm u_j}\,.
\end{equation}

\begin{definition}[Local overlap regularity (LOR)] \label{def:LOR}
Let the Hamiltonian $H_0$ be as in Assumption~\ref{ass:H0}. We say that an observable $A$ satisfying Assumption \ref{ass:stateandobs}~(ii) has the \emph{local overlap regularity (LOR) property} if and only if the eigenvector overlaps $\langle \bm u_j, A \bm u_j \rangle$ are approximately constant in the following sense: There exists a constant $\mathfrak{A} \in \R $
such that\footnote{\label{ftn:LOR}In fact, it is sufficient to assume that the overlaps are regular in $j$ in the sense that there exists a uniformly equicontinuous sequence $(\mathfrak{A}_N)_{N \in \N}$ of functions $\mathfrak{A}_N: I_{2 \Delta} \to \R$ such that 
\begin{equation*} 
	\langle \bm u_j, A \bm u_j \rangle = \mathfrak{A}_N(\mu_j) + \mathcal{O}\big(\cE_{{\rm LOR}}\big) \quad \text{for all} \quad j \in \N \quad \text{with} \quad \mu_j \in I_{2\Delta}\,. 
\end{equation*}
Alternatively, we could also assume that, for every fixed $\Delta > 0$ small enough, the overlaps $	\langle \bm u_j, A \bm u_j \rangle$ are well approximated by $\mathfrak{A}_N(\mu_j) $ in $\ell^p$-sense for some $p\geq1$, i.e.,
\begin{displaymath}
\frac1N\sum_{j:\mu_j\in I_{2\Delta}} \left| \langle \bm u_j, A \bm u_j \rangle-\mathfrak{A}_N(\mu_j)\right|^p \overset{N\rightarrow\infty}{\longrightarrow}0 \,. 
\end{displaymath}
The case $p=2$ is reminiscent of the so called \textit{weak ETH} studied in~\cite{BKL2010, DRecho}.}
\begin{equation} \label{eq:LORP}
\langle \bm u_j, A \bm u_j \rangle = \mathfrak{A} + \mathcal{O}\big(\cE_{{\rm LOR}}\big) \quad \text{for all} \quad j \in \N \quad \text{with} \quad \mu_j \in I_{2\Delta}\,, 
\end{equation}
where the error $\cE_{\rm LOR} = \cE_{{\rm LOR}}(\Delta, N)$ satisfies
\begin{equation*}
	\lim_{\Delta\rightarrow0}\lim_{N\rightarrow\infty} \cE_{{\rm LOR}} = 0\,. 
\end{equation*}
\end{definition}
The LOR property \eqref{eq:LORP} is satisfied, e.g., if $H_0$ satisfies the Eigenstate Thermalization Hypothesis \cite{Deutsch1991, Srednicki} (see also the discussion in \cite{DRrelax1}). For general systems, the ETH remains an unproven hypothesis. We remark, however, that it has been rigorously proven for a large class of mean-field random matrices $H_0$ (see \cite{ETHpaper, equipart}), including Wigner matrices and their deformations. 

The following corollary collects our  rigorous results on the relaxation formula \eqref{eq:keyeq}. The first two parts
(items (a) and (b) below) immediately follow  from Theorem \ref{thm:main}~(a) and (b). The third part, item (c), involves the LOR property of the observable $A$ and requires a separate argument, provided in Section~\ref{subsec:relaxproof}.

\begin{corollary}[Relaxation formula] \label{cor:relax}
	Under the assumptions and using the notations of Theorem~\ref{thm:main}, it holds that
\begin{equation} \label{eq:relaxform} 
	\langle A \rangle_{P_\lambda(t)} =  \langle A \rangle_{\widetilde{P}_{\lambda}}   + |  g_\lambda(t)|^2  \left[ \langle A\rangle_{P_0(t)}- \langle A \rangle_{\widetilde{P}_{\lambda}}\right] + \mathcal{R} + \mathcal{E}\,. 
\end{equation}
In particular, we have the following: 
\begin{itemize}
\item[(a)] {\emph{[Short kinetic time behavior]}} Let $0<T\lesssim1$. Then, it holds that 
\begin{equation} \label{eq:short}
\limsup_{\Delta\rightarrow0} \limsup_{\substack{t\rightarrow\infty\\\lambda\rightarrow0\\ \lambda^2 t=T}} \limsup_{N\rightarrow\infty} \left| 	\langle A \rangle_{P_\lambda(t)}  - \langle A\rangle_{P_0(t)}\right| \lesssim T \qquad \text{\emph{almost surely (a.s.)}}
\end{equation}

\item[(b)]{\emph{[Long kinetic time behavior]}} Let $T\gtrsim1$. Then it holds that (recall $\alpha = \pi \rho_0(E_0)$)
\begin{equation} \label{eq:long}
	\limsup_{\Delta\rightarrow0} \limsup_{\substack{t\rightarrow\infty\\\lambda\rightarrow0\\ \lambda^2 t=T}} \limsup_{N\rightarrow\infty} \left| 	\langle A \rangle_{P_\lambda(t)}  - \langle A \rangle_{\widetilde{P}_{\lambda}}\right| \lesssim T \ee^{-2 \alpha T} \qquad \text{\emph{a.s.}}
\end{equation}
\end{itemize}
Moreover, additionally assuming that $A$ has the LOR property from Definition \ref{def:LOR}, we have: 
\begin{itemize}
\item[(c)] \emph{[Intermediate kinetic times under LOR]} For every fixed $T \in (0,\infty)$ it holds that 
\begin{equation} \label{eq:LORRsmall}
	\lim_{\Delta\rightarrow0}\lim_{\substack{t\rightarrow\infty\\\lambda\rightarrow0\\ \lambda^2t=T}}\lim_{N\rightarrow\infty}\big[ |\mathcal{R}| + |\cE| \big]= 0  \qquad  \text{\emph{a.s.}}\,,
\end{equation}
i.e.~the relaxation formula \eqref{eq:keyeq} is valid at \emph{all} kinetic times $T \in (0,\infty)$. 
 \end{itemize}
\end{corollary}
Summarizing Corollary \ref{cor:relax}, we have that the relaxation formula \eqref{eq:keyeq} generally holds in the two limiting regimes (a) $|g_\lambda(t)|^2 \approx 0$ and (b) $|g_\lambda(t)|^2 \approx 1$, i.e.~$T \ll 1$ or $T \gg 1$, respectively. In between, \eqref{eq:keyeq} is valid under the additional assumption that $A$ has the LOR property from Definition \ref{def:LOR}, as this allows for the improved bound \eqref{eq:LORRsmall} on $\mathcal{R}$ compared to \eqref{eq:convergencetildeP}. However, \emph{without} this regularity assumption, only the bound \eqref{eq:convergencetildeP} (i.e.~\eqref{eq:short} and \eqref{eq:long}) can hold, which indicates that the relaxation formula \eqref{eq:keyeq} is \emph{not} generally valid for intermediate kinetic times $T \sim 1$. Indeed, it is easy to construct a counterexample. Finally, we remark that Corollary \ref{cor:relax}~(c) holds verbatim if the state $P$ satisfies the LOR condition instead of the observable $A$. This simply follows by inspecting the proof in Section~\ref{subsec:relaxproof}. 

\begin{remark} \label{rmk:relax}We have two further comments on Corollary \ref{cor:relax}. 
\begin{itemize}
\item[(i)] The relaxation formula \eqref{eq:relaxform} is in perfect agreement with the main result of Dabelow and Reimann, see \cite[Eq.~(16)]{DRrelax1}. In fact, the state $\widetilde{P}_\lambda$ defined in \eqref{eq:kernelstateapprox} agrees with $\widetilde{\rho}$ from \cite[Eq.~(16)]{DRrelax1}, named the "'washed out' descendant of the so-called diagonal ensemble" \cite{DRrelax1}. 
\item[(ii)] 
In fact, recalling \eqref{eq:specdecH0mainres}, the proof of Theorem \ref{thm:main}~(b) in Section \ref{subsec:step3} reveals that (see \eqref{eq:Rdef}) the error $\mathcal{R}$ in \eqref{eq:mainb} and \eqref{eq:relaxform} is given by
\begin{equation} \label{eq:Rexplicit}
	\mathcal{R} = \frac{1}{r} \sum_{j,k} \langle \bm u_j, A \bm u_j \rangle \langle \bm u_k, P \bm u_k \rangle \mathfrak{R}_{\lambda, t}(\mu_j - \mu_k)\,,
\end{equation}
where we denoted 
\begin{equation} \label{eq:r_def}
	r := \int_\R \Tr [\Im M_0(x + \ii \alpha \lambda^2)] \,  \langle \Im M_0(x + \ii \alpha \lambda^2)  \rangle_P \,  \rd x
\end{equation} and 
\begin{equation} \label{eq:tildephi}
	\mathfrak{R}_{\lambda, t}(u) := \pi\ee^{-2\alpha\lambda^2 t}\frac{2\alpha\lambda^2}{u^2 + (2\alpha\lambda^2)^2}\biggl(1-\cos(tu)-2\alpha\lambda^2 t \frac{\sin(tu)}{tu}\biggr).
\end{equation}
The explicit error term \eqref{eq:Rexplicit}--\eqref{eq:tildephi} is in precise agreement with the error term in \cite{DRrelax1}, see Eqs.~(17) and (18). In particular, assuming the LOR property \eqref{eq:LORP} for $A$ (or $P$), the smallness of $\mathcal{R}$ in \eqref{eq:LORRsmall} for all kinetic times $T \in (0,\infty)$ is a consequence of the fact that $\int_\R 	\mathfrak{R}_{\lambda, t}(u) \rd u = 0$.
\end{itemize}
\end{remark}

\subsection{Prethermalization} \label{subsec:pretherm}
In this section, we specialize the general relaxation theory of perturbed quantum dynamics from Theorem \ref{thm:main} to a class of unperturbed Hamiltonians $H_0$ and states $P$ which \emph{has the prethermalization property} in the following sense. 
\begin{definition}[Prethermalization property]\label{def:preTcond}
Let the Hamiltonian $H_0$ and the state $P$ be defined as in Assumptions~\ref{ass:H0} and~\ref{ass:stateandobs}. We say that $(H_0,P)$ has the \emph{prethermalization property} if and only if there exists a state $P_{\rm pre}$ (called the \emph{prethermal state}) such that we have the following: 
\begin{subequations}\label{eq:prethermprop}
	\begin{itemize}
\item[(a)] The unperturbed time evolution $P_0(t)$ converges to $P_{\rm pre}$, i.e.,~for all\footnote{A common variant of this requirement in the physics literature is to assume the validity of \eqref{eq:prethermpropa} only for \emph{local} observables, i.e.~supported in a finite region of an underlying space (see, e.g., \cite[Section~5.2]{MIKUreview}).} observables $A$ satisfying Assumption \ref{ass:stateandobs}~(ii), it holds that 
\begin{equation} \label{eq:prethermpropa}
	\lim_{t\rightarrow\infty} \lim\limits_{N \to \infty}\left[\langle A\rangle_{P_0(t)}- \langle A \rangle_{P_{\mathrm{pre}}}\right] = 0 \,. 
\end{equation} 
\item[(b)] There exists an observable $A_0$ (satisfying Assumption \ref{ass:stateandobs}~(ii)) which distinguishes $P_{\rm pre}$ from $\widetilde{P}_\lambda$ (cf.~\eqref{eq:kernelstateapprox}), i.e.~there exists a constant $\mathfrak{c}_{\rm pre}>0$ such that 
\begin{equation} \label{eq:prethermpropb}
\liminf_{ \lambda \to 0} \liminf_{N \to \infty} \left| \langle A_0 \rangle_{P_{\mathrm{pre}}} - \langle A_0 \rangle_{\widetilde{P}_\lambda} \right| \ge \mathfrak{c}_{\rm pre}\,. 
\end{equation} 
	\end{itemize}
\end{subequations}

\end{definition}
We emphasize that $(H_0, P)$ having the \emph{prethermalization property} is a purely deterministic condition, i.e., in particular, it does not depend on the Wigner matrix $W$. In the physics literature (see, e.g., \cite{MRdR2019} but also \cite{MIKUreview, LGS2016, KWE2011, BalzReimann}), the prethermalization property is generally expected to be satisfied if $H_0$ is an \emph{integrable} Hamiltonian having at least one additional conserved quantity $Q$ for which $[H_0, Q] = 0$.\footnote{We use the usual notation for the commutator, i.e.~$[A,B] := AB - BA$ for all matrices $A,B$.} This symmetry is then broken by a generic perturbation $W$, i.e.~$[W,Q] \neq 0$. 
 In the presence of $M$ conserved quantities $(Q_k)_{k=1}^M$, a good candidate for the prethermal state $P_{\rm pre}$ is given by the so called \emph{generalized Gibbs ensemble (GGE)}
$$
P_{\rm GGE} = \frac{\ee^{- \sum_{k=1}^M \lambda_k Q_k}}{\Tr \, \ee^{- \sum_{k=1}^M \lambda_k Q_k}}\,,
$$
where the parameters $\lambda_k$ are chosen in such a way that $\Tr Q_k P_{\rm GGE} = \Tr Q_k P$ for all $k\in [M]$ (see, e.g., \cite{KWE2011} and \cite[Section~5.1]{MIKUreview}). 
 \nc
Exemplary pairs $(H_0, P)$ and observables $A_0$ satisfying the conditions in Definition \ref{def:preTcond} are given in Section \ref{subsec:example}.

 {Assuming that $(H_0, P)$ has the prethermalization property, Theorem~\ref{thm:main} reads as follows.}

\begin{corollary}[Prethermalization]\label{cor:PreT1}
Under the assumptions of Theorem~\ref{thm:main}, let further $(H_0,P)$ have the prethermalization property from Definition \ref{def:preTcond}. Then, recalling the notations from Theorem \ref{thm:main}, it holds that
\begin{equation} \label{eq:preTregimes}
	\langle A \rangle_{P_\lambda(t)} =  \langle A \rangle_{\widetilde{P}_{\lambda}}   + |  g_\lambda(t)|^2  \left[ \langle A\rangle_{P_{\rm pre}}- \langle A \rangle_{\widetilde{P}_{\lambda}}\right] + \mathcal{R} + \mathcal{E}'\,. 
\end{equation}
We have $\cE' = \mathcal{O}(\cE_0') + \mathcal{O}_\prec (C(\lambda, t)/\sqrt{N})$ for some constant $C(\lambda, t) > 0$ and for every fixed $T \in (0, \infty)$, the deterministic errors $\cE_0'=\cE_0'(\lambda, t, \Delta, N)$ and $\mathcal{R} = \mathcal{R}(\lambda, t, \Delta, N)$ satisfy
\begin{equation*} 
	\lim_{\Delta\rightarrow0}\lim_{\substack{t\rightarrow\infty\\\lambda\rightarrow0\\ \lambda^2t=T}}\lim_{N\rightarrow\infty} \cE_0'=0 \qquad \text{and} \qquad 		\limsup_{\Delta\rightarrow0} \limsup_{\substack{t\rightarrow\infty\\\lambda\rightarrow0\\ \lambda^2 t=T}} \limsup_{N\rightarrow\infty} |\mathcal{R}| \lesssim T\ee^{-2\alpha T}\,.
\end{equation*}

\end{corollary}
We remark that the error term $\cE_0'$ contributing in \eqref{eq:preTregimes} consists of two parts, $\cE_0' = \cE_0 + \cE_{\rm pre}$, with $\cE_0 = \cE_0(\lambda, t, \Delta, N)$ from Theorem \ref{thm:main} and $\cE_{\rm pre} = \cE_{\rm pre}(t,N)$ being the (absolute value of the) error in \eqref{eq:prethermpropa}. 
Note that~\eqref{eq:preTregimes} in particular implies the following small and large $T$ behaviors:
\begin{equation} \label{eq:pretherm}
	\begin{split}
	\limsup_{\Delta\rightarrow0} \limsup_{\substack{t\rightarrow\infty\\\lambda\rightarrow0\\ \lambda^2 t=T}} \limsup_{N\rightarrow\infty}\left| \langle A\rangle_{P_\lambda(t)} - \langle A\rangle_{P_{\rm pre}}\right| &\lesssim T \quad \text{for} \quad T \lesssim1 \quad \text{{a.s.}}\,, \\[2mm]
\text{and} \quad \	\limsup_{\Delta\rightarrow0} \limsup_{\substack{t\rightarrow\infty\\\lambda\rightarrow0\\ \lambda^2 t=T}} \limsup_{N\rightarrow\infty}\left| \langle A\rangle_{P_\lambda(t)}-\langle A\rangle_{\widetilde{P}_{\lambda}} \right|&\lesssim T\ee^{-2\alpha T} \quad \text{for} \quad T \gtrsim 1 \quad \text{{a.s.}}
	\end{split}
\end{equation}
Moreover, \eqref{eq:prethermpropb} ensures that $\langle A \rangle_{P_{\mathrm{pre}}} \neq \langle A \rangle_{\widetilde{P}_\lambda}$ for at least one observable $A= A_0$, which, together with \eqref{eq:pretherm} establishes Figure \ref{fig-2plateau} as a schematic graph of a prethermalization process.

\subsection{Connection to the microcanonical ensemble} \label{subsec:mcensemble}
Under an additional regularity assumption on $x \mapsto \llangle \Im M_0(x + \ii \alpha \lambda^2 ) A \rrangle$ we can relate the state $\widetilde{P}_\lambda$ from \eqref{eq:kernelstateapprox} to the microcanonical ensemble.

\begin{theorem}[Microcanonical average]\label{thm:PreT2}
	Under the assumptions of Theorem~\ref{thm:main},
	let us further assume that 
	\begin{equation*}
		h \equiv h(\lambda, N) : x \mapsto \llangle \Im M_0(x + \ii \alpha \lambda^2 ) A \rrangle
	\end{equation*}
	is a Lipschitz continuous map on $I_\Delta$ with Lipschitz constant $\mathrm{Lip}_{I_\Delta} (h)$ bounded in the sense that
	\begin{equation} \label{eq:hC1 norm}
		\limsup_{\Delta \to 0} \, \limsup_{\lambda \to 0} \, \limsup_{N \to \infty}\mathrm{Lip}_{I_\Delta} (h) \lesssim 1\,. 
	\end{equation}
	Then
	\begin{equation}\label{eq:mc}
		\langle A\rangle_{\widetilde{P}_\lambda}= \langle A \rangle_{P^{\rm (mc)}_{\lambda}}+\cE_{\rm mc} \qquad \text{with} \qquad P^{\rm (mc)}_\lambda := \frac{\Im M_0(E_0+ \ii \alpha \lambda^2)}{\Tr[\Im M_0(E_0 + \ii \alpha \lambda^2)]} \,, 
	\end{equation}
	where the error $\cE_{\rm mc}=\cE_{\rm mc}(\lambda, \Delta, N)$ satisfies 
	\begin{displaymath}
		\lim_{\Delta\rightarrow0}\lim_{\substack{\lambda\rightarrow0}}\lim_{N\rightarrow\infty}\cE_{\rm mc}=0.
	\end{displaymath}
\end{theorem}
We emphasize that $P^{\rm (mc)}_\lambda$ is completely independent of the initial state $P$. Moreover, as mentioned above and already indicated by the notation, we can interpret $\langle A \rangle_{P_\lambda^{\rm (mc)}}$ from~\eqref{eq:mc} as the microcanonical average of $H_\lambda$ at energy $E_0$. The reason underlying this interpretation is that for any normalized eigenvector $\bm v_{\lambda}$ of $H_\lambda$ with eigenvalue $E_\lambda$ very close to $E_0$, it holds that\footnote{Indeed, taking, say, $E_\lambda \in I_{\Delta/2}$ it can be rigorously shown that the difference $\langle \bm v_{\lambda}, A \bm v_{\lambda} \rangle- \langle A \rangle_{P^{\rm (mc)}_{\lambda}}$ vanishes in the triple limit \eqref{eq:triplelimvanish}. More precisely, this follows from the \emph{Eigenstate Thermalization Hypothesis (ETH)} for the random matrix $H_\lambda = H_0 + \lambda W$ (see \cite[Theorem 2.6]{equipart}) and using assumption \eqref{eq:hC1 norm} together with \eqref{eq:MDE} and \eqref{eq:trM_est}.}
\begin{equation*}
	\langle \bm v_{\lambda}, A \bm v_{\lambda} \rangle \approx 
	\frac{\Tr [\Im M_0(E_0+ \ii \alpha \lambda^2)A ]}{\Tr[\Im M_0(E_0 + \ii \alpha \lambda^2)]}  =\langle A \rangle_{P^{\rm (mc)}_{\lambda}} \,. 
\end{equation*}
This means, $P^{\rm (mc)}_\lambda$ is a close effective approximation to the actual projection $ \ket{\bm v_{\lambda}} \bra{\bm v_{\lambda}}$ onto the eigenspace spanned by $\bm v_{\lambda}$.

\subsection{Examples} \label{subsec:example}
In this section, we give two examples of physical settings where prethermalization occurs and connect them to our assumptions. Note that both examples are one-dimensional, however, the extension to higher dimensions is straightforward.

\subsubsection{Next-nearest neighbor hopping}
For $N \in \N$ even, we consider the Laplacian-like Hamiltonian $H_0$ acting on functions $\psi\in\ell^2(\Z/(N\Z))$ as
\begin{equation}\label{eq:NNNHamiltonian}
(H_0\psi)(x):=2\psi(x)-\psi(x-2)-\psi(x+2)
\end{equation}
where $x-2$ and $x+2$ are interpreted mod $N$. Note that $H_0$ is similar to the discrete Laplacian with periodic boundary condition but induces next-nearest neighbor hopping instead of nearest neighbor hopping. In particular, $H_0$ conserves parity in the sense that functions that are only supported on the even or odd points of $\Z/(N\Z)$, respectively, remain invariant, and thus its spectrum has an additional two-fold 
degeneracy. Similar to the routine computations done for the discrete Laplacian, one can readily check the following: 
\begin{itemize}
\item The Hamiltonian $H_0$ is bounded, $\|H_0\|\lesssim 1$. 
\item Its spectrum is given by $\sigma(H_0) = \left\{ 2 (1 - \cos(2 p_j)) : p_j = 2 \pi j /N \right\}_{j \in [N]} \subset [0,4]$. 
\item The limiting density of states as $N\rightarrow\infty$ evaluates to
\begin{equation} \label{eq:DOS}
\rho_0(x)=\frac{1}{\pi\sqrt{x(4-x)}}\mathds{1}_{[0,4]}(x)
\end{equation}
which is compactly supported and satisfies the regularity assumptions in Assumption~\ref{ass:H0} for $x$ bounded away from $0$ and $4$.  
\end{itemize}
In this setting, we fix $k$ such that the eigenvalue $2(1 - \cos(2 p_k))$ satisfies $p_k\in(0,\pi/2)$. Now take  $P:= \ket{\bm u_k} \bra{\bm u_k}$ with $\bm u_k$ being the normalized eigenvector of $H_0$ supported on the \emph{even} sub-lattice corresponding to the eigenvalue $2(1 - \cos(2 p_k))$. By construction, for every observable $A$ satisfying Assumption \ref{ass:stateandobs}, we have
\begin{displaymath}
\langle A\rangle_{P_0(t)}= \langle A \rangle_{P} = \langle A \rangle_{P_{\rm pre}}\,, \quad \text{for all} \quad t\geq0\,,
\end{displaymath}
since $[P,H_0]=0$. Hence, the symmetry implies that $P_{\rm pre} = P$. In particular, for $A:=\mathbf{1}_{\rm odd}$ being the identity operator on the \emph{odd} sub-lattice, its prethermal value is given by $\langle A\rangle_{P_{\rm pre}}=0$. Moreover, by spectral decomposition of $H_0 = \sum_j \mu_j \ket{\bm u_j} \bra{\bm u_j}$, we obtain
\begin{equation} \label{eq:specdecexample}
\llangle\Im M_0(x+\ii\lambda^2\alpha)A\rrangle=\frac{1}{N}\sum_j\langle\bm{u}_j,A\bm{u}_j\rangle\frac{\alpha\lambda^2}{|x-\mu_j|^2+(\alpha\lambda^2)^2}>0,
\end{equation}
which implies that $\langle A\rangle_{\widetilde{P}_\lambda}\neq\langle A\rangle_{P_{\rm pre}}$ (recall the definition of $\widetilde{P}_\lambda$ in \eqref{eq:kernelstateapprox}) for $A = \mathbf{1}_{\rm odd}$. Hence, we deduce that $(H_0, P)$ has the prethermalization property from Definition~\ref{def:preTcond}.

\subsubsection{Free spinless fermions on a lattice}
As our second example, we consider a model of spinless fermions in a periodic one-dimensional lattice of even length $N$ (cf.~\cite[App.~B]{DRrelax1}), which can be seen as a many-body analog of the first example (although with nearest neighbor hopping instead of next-nearest neighbor hopping). Let
\begin{equation}\label{eq:FFHamiltonian}
H_0=\frac{1}{\sqrt{N}}\sum_j c_j^{\dagger}c_{j+1}+c_{j+1}^\dagger c_j,
\end{equation}
where $c_j^\dagger$ and $c_j$ denote the fermionic creation and annihilation operators at site $j$, and the summation indices are considered modulo $N$. Note that the Hamiltonian in~\eqref{eq:FFHamiltonian} conserves the particle number. It is readily checked that $H_0$ admits a limiting density of states which is not compactly supported but has fast decaying (Gaussian) tails. As the regularity assumptions in Assumption~\ref{ass:H0}~(ii) are satisfied, this example is still sufficiently close to our theory to be described by it reasonably well. 

In this setting, pick $\psi_j$ as the orthonormal eigenfunctions of the discrete Laplacian describing nearest neighbor hopping with periodic boundary conditions (i.e.~the analog of \eqref{eq:NNNHamiltonian} with $\pm 1$ instead of $\pm 2$) corresponding to the eigenvalues $2(1 - \cos(p_j))$ with 
\begin{displaymath}
p_j :=\frac{2 \pi j}{N},\quad \frac{j}{N} \in \Big[\frac18, \frac38\Big] \cup \Big[\frac58, \frac78\Big].
\end{displaymath}
We then construct $P:=\ket{\psi} \bra{\psi}$ as a rank$-1$ projection onto an eigenstate of $H_0$ by taking
\begin{displaymath}
\psi:=\bigwedge_{j}\psi_j
\end{displaymath}
as a Slater determinant of the $N/2$ one-particle wave functions $\psi_j$. This ensures that $P$ satisfies Assumption \ref{ass:stateandobs}~(i), as the density of states (which is the same as \eqref{eq:DOS}) is regular in such intervals. Noting that $[P,H_0]=0$, we obtain
\begin{displaymath}
\langle A\rangle_{P_0(t)}= \langle A \rangle_{P} = \langle A \rangle_{P_{\rm pre}}\,, \quad \text{for all} \quad t\geq0\,,
\end{displaymath}
for every observable $A$ satisfying Assumption~\ref{ass:stateandobs}. Hence, $P_{\rm pre} = P$, similar to the first example. In particular, for ${A=\mathbf{1}_{\mathcal{H}_{N/2}^{\perp}}}$ being the identity on the orthogonal complement of the $N/2$-particle sector of the Fock space, the prethermal value is given by $\langle A\rangle_{P_{\rm pre}} = 0$. Moreover, by spectral decomposition of $H_0$, similarly to \eqref{eq:specdecexample},  we find that $\langle A\rangle_{\widetilde{P}_\lambda} \neq \langle A\rangle_{P_{\rm pre}}$. Hence, we deduce that $(H_0, P)$ has the prethermalization property from Definition~\ref{def:preTcond}.

\section{Proofs} \label{sec:proofs}
In this section, we provide the proofs of our main results formulated in Section \ref{sec:main}. 
We begin by giving the proof of Theorem~\ref{thm:main}, which we organize in three steps:
\begin{itemize}
\item[(i)] In Section \ref{subsec:step1}, as the first step, we approximate the random time evolution $\langle A \rangle_{P_\lambda(t)}$ by a deterministic object, up to an error vanishing as $N \to \infty$
with very high probability. 
 This is done using a suitable \emph{global law} for two resolvents of the random matrix $H_\lambda$ (see Proposition \ref{prop:LL} below). 
\item[(ii)] The deterministic object resulting from Step (i) consists of two terms, a regular and a singular one. In Section \ref{subsec:step2}, we evaluate these terms up to errors captured by $\mathcal{E}$ (see Proposition \ref{prop:regsing}). This proves Theorem \ref{thm:main}~(a). 
\item[(iii)] As the third and final step in Section \ref{subsec:step3}, we examine the behavior of the singular term in the limits $T \to 0$ and $T \to \infty$ for $T := \lambda^2 t$ (see Proposition \ref{prop:kernel}). This proves Theorem~\ref{thm:main}~(b).  
\end{itemize}
Afterwards, we give the proofs of Corollary \ref{cor:relax} and Theorem~\ref{thm:PreT2} in Sections \ref{subsec:relaxproof} and \ref{subsec:proofmc}, respectively. The proof of Corollary~\ref{cor:PreT1} is immediate from Definition~\ref{def:preTcond} and Theorem~\ref{thm:main} and hence omitted.

\subsection{Step (i): Global law} \label{subsec:step1} Ignoring $\lambda$ for a minute, let $H := D + W$ such that $D\in\C^{N\times N}$ is a self-adjoint deterministic matrix with $\|D\|\lesssim1$ and $W$ is a Wigner matrix satisfying Assumption~\ref{ass:W}. We refer to $H$ as a \emph{deformed Wigner matrix}. It is well known  \cite{AEK2019, slowcorr, AEK2020, edgelocallaw}, that the resolvent  $G(z):= (H-z)^{-1}$ of $H$ at spectral parameter $z \in \C\setminus \R$ is very well approximated by a deterministic matrix $M$, which is the unique solution to the \emph{Matrix Dyson Equation (MDE)}\footnote{\label{ftn:MDE}The MDE in the context of mean-field random matrices was introduced in \cite{AEK2019} and extensively analyzed in \cite{AEK2020}.} 
\begin{equation} \label{eq:MDEnolambda}
	-\frac{1}{M(z)}=z-D+\llangle M(z)\rrangle \,, \quad \text{with} \quad \Im M(z) \Im z > 0\,. 
\end{equation}
In particular (see \cite[Theorem~2.1]{slowcorr}), for $|\Im z| \gtrsim 1$ and arbitrary deterministic matrix $B \in \C^{N \times N}$ with $\|B\|\lesssim1$, it holds that 
\begin{equation} \label{eq:singlegloballaw}
\big|\llangle (G(z) - M(z))B \rrangle\big| \prec \frac{1}{N}\,. 
\end{equation}

For our purposes, it is not sufficient to approximate only a single resolvent in the sense of \eqref{eq:singlegloballaw}. Instead, we need to establish the deterministic approximation to $\langle \bm x, G(z_1) B G(z_2) \bm y \rangle$ with two deterministic vectors $\bm x, \bm y$. This is the content of the following proposition, the proof of which is given in Appendix \ref{app:LL}.

\begin{proposition}[Isotropic two-resolvent global law for deformed Wigner matrices]\label{prop:LL}
Let $H := D + W$ be an $N\times N$ deformed Wigner matrix
 (as in Assumption \ref{ass:W}) with a bounded self-adjoint deformation $D\in\C^{N\times N}$. Pick $B\in\C^{N\times N}$ a deterministic matrix with $\|B\|\lesssim1$, deterministic vectors $\bm x,\bm y\in\C^N$ with $\|\bm x \|=\| \bm y\|=1$ as well as $z_1,z_2\in\C$ with $\min\{ |\Im z_1|,| \Im z_2|\}\ge1$.\footnote{As used also in \eqref{eq:LLlambda} below, the constant one can be replaced by any other $N$-independent $c > 0$.} Denote further $G_i=G(z_i)=(H-z_i)^{-1}$. Then,
\begin{equation}\label{eq:LL}
\Big|\langle\bm x,G_1BG_2\bm y\rangle-\Big\langle \bm x,\Big(M_1BM_2+\frac{M_1M_2 \llangle M_1BM_2 \rrangle }{1- \llangle M_1M_2 \rrangle}\Big)\bm y\Big\rangle\Big|\prec\frac{1}{\sqrt{N}}\,,
\end{equation}
where we denoted $M_i = M(z_i)$ with $M(z) \in \C^{N \times N}$ being the solution of \eqref{eq:MDEnolambda}.
\end{proposition}

We now apply Proposition \ref{prop:LL} to resolvents $G_\lambda(z) := (H_\lambda -z)^{-1}$ of our concrete deformed Wigner random matrix $H_\lambda = H_0 + \lambda W$. First, the scaling by $\lambda$ leads to the following version of the MDE \eqref{eq:MDEnolambda}: 
\begin{equation} \label{eq:MDE}
	- \frac{1}{M_\lambda(z)} = z - H_0 + \lambda^2 \llangle M_\lambda(z) \rrangle, \quad \Im M_{\lambda}(z)\Im z >0\,. 
\end{equation}
Next, inspecting the proof of Proposition \ref{prop:LL} in Appendix \ref{app:LL}, we find that \eqref{eq:LL} becomes
\begin{equation}\label{eq:LLlambda}
	\begin{split}
	\Bigg|\big(&G_\lambda(z_1)BG_\lambda(z_2)\big)_{\bm x \bm y} \\
	 &-\left(M_\lambda(z_1)BM_\lambda(z_2)+\lambda^2\frac{M_\lambda(z_1)M_\lambda(z_2) \llangle M_\lambda(z_1)BM_\lambda(z_2) \rrangle}{1- \lambda^2\llangle M_\lambda(z_1)M_\lambda(z_2) \rrangle}\right)_{\bm x \bm y}\Bigg|\prec\frac{C(\lambda,c)}{\sqrt{N}}\,,
	\end{split}
\end{equation}
for $z_1, z_2 \in \C \setminus \R$ with $\min\{ |\Im z_1|,| \Im z_2|\}\ge c $ for some $c > 0$. The positive constant $C(\lambda,c)$ in \eqref{eq:LLlambda} depends only on its arguments $\lambda$ and $c$.\footnote{In fact, by examining the proof of Proposition \ref{prop:LL}, it can easily be seen that the dependence on both 
small parameters $\lambda$ and $c$ is at most (inverse) polynomial.}

For the proof of Theorem~\ref{thm:main}, we now employ \eqref{eq:LLlambda} as follows: Applying residue calculus allows to rewrite $\langle A\rangle_{P_\lambda(t)}$ as the contour integral
\begin{equation}\label{eq:contourint}
\langle A\rangle_{P_\lambda(t)}=\langle \ee^{-\ii t H_\lambda}P\ee^{\ii t H_\lambda}A\rangle= \frac{1}{(2 \pi \ii)^2 }\oint_{\gamma_1}\oint_{\gamma_2}\ee^{\ii t(z_1-z_2)}\Tr\big[ G_{\lambda}(z_1)AG_{\lambda}(z_2)P\big]\rd z_1\rd z_2
\end{equation}
where $\gamma_1$ and $\gamma_2$ are two semicircles (each being the complex conjugate of the other) with some large radius $ R  \gtrsim 1$ (see~Figure \ref{fig:contours} below). We further define the contours such that the distance between the flat part of the semicircles and the real line is~$t^{-1}$. Note that we have
\begin{equation}\label{eq:specinclusion}
\sigma(H_\lambda)\subseteq\sigma(H_0) + [-(2+\epsilon)\lambda, (2+\epsilon) \lambda]\quad \text{w.v.h.p.}
\end{equation}
for any fixed $\epsilon>0$ by standard perturbation theory, using that $\Vert W \Vert \le 2 + \epsilon$ with very high probability (see, e.g.,~\cite[Theorem~7.6]{semicirclegeneral}). In particular, the contours encircle the spectrum of $H_\lambda$ completely if $R$ is chosen large enough.

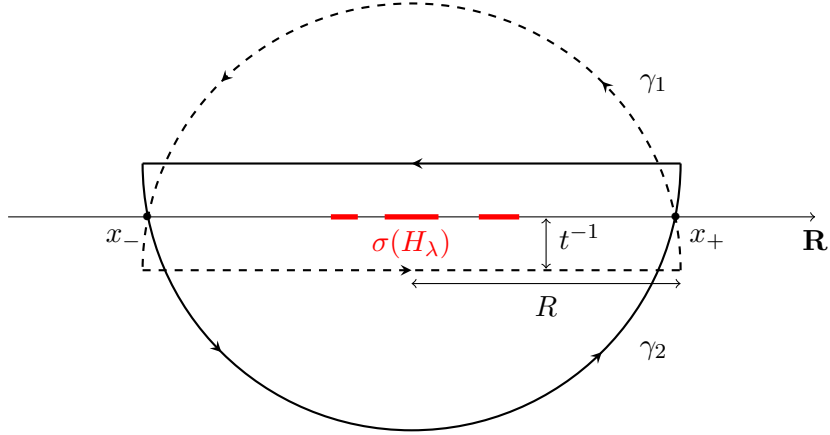
\begin{figure}[h]
\begin{center}
\begin{tikzpicture}[scale=\textwidth/21.2cm]
\nc
\draw[black,->] (-5,0) -- (10,0);
\draw (10,0) node[below=1pt] {\nc$\R$};

\draw[red,line width=0.7mm] (1,0) -- (1.5,0);
\draw[red,line width=0.7mm] (2,0) -- (3,0);
\draw[red,line width=0.7mm] (3.75,0) -- (4.5,0);
\draw (2.5,0) node[below=1pt] {\color{red}$\sigma(H_\lambda)$\nc};

\draw (6.5,2.5) node[right=1pt] {\nc $\gamma_1$};
\draw (6.5,-2.5) node[right=1pt] {\nc $\gamma_2$};

\draw[black,thick,dashed,postaction={on each segment={mid arrow=black}}] (7.5,-1) arc[start angle=0, end angle=180, radius=5];
\draw[black,thick,dashed,postaction={on each segment={mid arrow=black}}] (-2.5,-1) -- (7.5,-1);

\draw[black,<->] (2.5,-1.25) -- (7.5,-1.25);
\draw (5,-0.375) node[right=1pt] {\nc $ t^{-1}$};
\draw[black,<->] (5,-0.97) -- (5,-0.03);
\draw (5,-1.25) node[below=1pt] {\nc $R$};
\draw (8,0) node[below=1pt] {\nc $x_+$};
\draw (-2.85,0) node[below=1pt] {\nc $x_-$};
\node at (-2.41,-.05) {\textbullet};
\node at (7.41,-.05) {\textbullet};

\draw[black,thick,postaction={on each segment={mid arrow=black}}] (-2.5,1) arc[start angle=180, end angle=360, radius=5];
\draw[black,thick,postaction={on each segment={mid arrow=black}}] (7.5,1) -- (-2.5,1);
\end{tikzpicture}
\end{center}
\captionof{figure}{Sketch of the contours $\gamma_1$ and $\gamma_2$ from \eqref{eq:contourint}. The intersections of $\gamma_1$ (and $\gamma_2$) with the real line are denoted by $x_\pm$.}\label{fig:contours}
\end{figure} 

Writing $P=\sum_jp_j \ket{\bm p_j} \bra{\bm p_j}$ in spectral decomposition and using that the $p_j \in [0,1]$ sum to one by Assumption~\ref{ass:stateandobs}~(i), the global law~\eqref{eq:LLlambda}
applied to ${\bm x}={\bm y}={\bm p}_j$ implies\footnote{Note that the true bound in the averaged law is of order $N^{-1}$, i.e., by a factor $N^{-1/2}$ better than~\eqref{eq:globallambda}. However, for the following applications, our weaker bound is sufficient.}
\begin{equation} \label{eq:globallambda}
\begin{split}
\Bigg|\Tr[ &G_{\lambda}(z_1)AG_{\lambda}(z_2)P] \\
&-\Tr\left[M_\lambda(z_1)AM_\lambda(z_2)P+\lambda^2\frac{M_\lambda(z_1)M_\lambda(z_2)P \llangle M_\lambda(z_1)AM_\lambda(z_2) \rrangle}{1- \lambda^2\llangle M_\lambda(z_1)M_\lambda(z_2) \rrangle} \right]\Bigg|\prec\frac{C(\lambda, t)}{\sqrt{N}}
\end{split}
\end{equation}
uniformly for $z_1, z_2$ along the contours $\gamma_1, \gamma_2$ for any fixed $\lambda>0$. Just as in \eqref{eq:LLlambda}, $C(\lambda, t)$ denotes a positive constant depending only on $\lambda$ and $t$. Therefore, combining \eqref{eq:contourint} with \eqref{eq:globallambda}, we find that
\begin{equation} \label{eq:terms}
	\begin{split}
\langle A\rangle_{P_\lambda(t)}&=\frac{1}{(2 \pi \ii)^2 }\oint_{\gamma_1}\oint_{\gamma_2} \ee^{\ii t (z_1 - z_2)}\Tr \big[M_{\lambda}(z_1)AM_{\lambda}(z_2)P \big] \rd z_1\rd z_2\\
&\quad +\frac{1}{(2 \pi \ii)^2 }\oint_{\gamma_1}\oint_{\gamma_2}  \ee^{\ii t (z_1 - z_2)}\lambda^2\frac{\llangle[\bigl]M_{\lambda}(z_1)AM_{\lambda}(z_2)\rrangle \Tr \big[M_{\lambda}(z_1)PM_{\lambda}(z_2)\big]}{1-\lambda^2\llangle[\bigl]M_{\lambda}(z_1)M_{\lambda}(z_2)\rrangle[\bigr]} \rd z_1\rd z_2 \\
& \quad +\cO_\prec\left(\frac{C(\lambda, t)}{\sqrt{N}}\right)\,.
	\end{split}
\end{equation}
To establish~\eqref{eq:main}, our main task thus lies in evaluating the right-hand side of~\eqref{eq:terms}. For simplicity, we refer to the integrals in the first and second line of~\eqref{eq:terms} as the \textit{regular} and the \textit{singular} term, respectively.

\subsection{Step (ii): Evaluation of the regular and singular term and proof of Theorem \ref{thm:main}~(a)} \label{subsec:step2}
We organize the result of our computation of \eqref{eq:terms} in the following proposition. 
\begin{proposition}[Evaluation of the regular and singular term]\label{prop:regsing}
Under the assumptions of Theorem~\ref{thm:main} and letting $\gamma_1$, $\gamma_2$ be the contours in Fig.~\ref{fig:contours}, we have (recalling $\alpha = \pi \rho_0(E_0)$)
\begin{subequations}
\begin{equation}\label{eq:regterm}
	\begin{split}
	\frac{1}{(2 \pi \ii)^2 }\oint_{\gamma_1}\oint_{\gamma_2} \ee^{\ii t (z_1 - z_2)} & \Tr \big[M_{\lambda}(z_1)AM_{\lambda}(z_2)P \big] \rd z_1\rd z_2 \\[1mm]
	&=\ee^{-2 \alpha \lambda^2t}\langle A\rangle_{P_0(t)}+\cO(\cE_{\rm reg})\,, 
	\end{split}
\end{equation}
for the regular term and
\begin{equation}\label{eq:singterm}
	\begin{split}
		\tilsing:=	\frac{\lambda^2}{(2 \pi \ii)^2 }\oint_{\gamma_1}\oint_{\gamma_2}  \ee^{\ii t (z_1 - z_2)} &\frac{\llangle[\bigl]M_{\lambda}(z_1)AM_{\lambda}(z_2)\rrangle \Tr \big[M_{\lambda}(z_1)PM_{\lambda}(z_2)\big]}{1-\lambda^2\llangle[\bigl]M_{\lambda}(z_1)M_{\lambda}(z_2)\rrangle[\bigr]} \rd z_1\rd z_2 \\[1mm]
		&=\langle A\rangle_{\widetilde{P}_{\lambda,t}}+\cO(\cE_{\rm sing})
	\end{split}
\end{equation}
\end{subequations}
for the singular term, with some error terms $\cE_{\rm reg/sing} = \cE_{\rm reg/sing}(\lambda, t, \Delta, N)$ in \eqref{eq:regterm} and~\eqref{eq:singterm} satisfying \eqref{eq:triplelimvanish}. The explicit form of $\cE_{\rm reg}$ is given in \eqref{eq:Ereg} and \eqref{eq:Eregnew}, while the explicit form of $\cE_{\rm sing}$ is given in \eqref{eq:Fsing_replace}. 

\end{proposition}
Plugging \eqref{eq:regterm} and~\eqref{eq:singterm} into \eqref{eq:terms}, we immediately conclude Theorem \ref{thm:main}~(a) after setting $\cE_0 := \cE_{\rm reg} + \cE_{\rm sing}$ and including the error term from \eqref{eq:terms} into $\cE$. \qed

It thus remains to give the proof of Proposition \ref{prop:regsing}, i.e.~its two parts \eqref{eq:regterm} and \eqref{eq:singterm}. This is done in Sections \ref{subsubsec:regterm} and \ref{subsubsec:singterm}, respectively. 

\subsubsection{Proof of \eqref{eq:regterm}} \label{subsubsec:regterm}
The main contribution to the integral in \eqref{eq:regterm} comes from the regime\footnote{Recall that the flat pieces of the contours $\gamma_1$ and $\gamma_2$ from Figure \ref{fig:contours} lie on the lower and upper half plane, respectively.} where $z_1$ and $z_2$ are close to $E_0$. Hence, as a first approximation we use the replacements $\llangle M_\lambda(z_1) \rrangle \approx \overline{m_0(E_0)}$ and $\llangle M_\lambda(z_2) \rrangle \approx {m_0(E_0)}$ in \eqref{eq:MDE}, which leads to
\begin{equation} \label{eq:Mapprox}
M_\lambda(z_1)  \approx \frac{1}{H_0 -z_1 - \lambda^2 \overline{m_0(E_0)}} \qquad \text{and} \qquad M_\lambda(z_2)  \approx \frac{1}{H_0 -z_2 - \lambda^2 {m_0(E_0)}}\,.
\end{equation}
{Applying the replacements in \eqref{eq:Mapprox} for the term in~\eqref{eq:regterm} yields}
\begin{equation} \label{eq:regcompute}
\begin{split}
	&\frac{1}{(2 \pi \ii)^2 }\oint_{\gamma_1}\oint_{\gamma_2} \ee^{\ii t (z_1 - z_2)}\Tr \left[\frac{1}{H_0 -z_1 - \lambda^2 \overline{m_0(E_0)}}A\frac{1}{H_0 -z_2 - \lambda^2 {m_0(E_0)}}P \right] \rd z_1\rd z_2 \\[2mm]
	= &\Tr \left[ \ee^{\ii t(H_0 - \lambda^2 \overline{m_0(E_0)})}A \ee^{- \ii t(H_0 - \lambda^2 {m_0(E_0)})} P \right] \\[2mm]
	= &\ee^{-2 \Im m_0(E_0)\lambda^2 t} \Tr \left[ \ee^{\ii tH_0 }A \ee^{- \ii tH_0 } P \right] = \ee^{-2 \alpha\lambda^2 t} \,  \langle A \rangle_{P_0(t)},
\end{split}
\end{equation}
since $\Im m_0(E_0) = \pi \rho_0(E_0)$, from simple residue calculus for $\lambda>0$ small enough, using that $|m_0(E_0)| \lesssim 1$ and $\gamma_1, \gamma_2$ encircle the spectrum of $H_0$. We have thus extracted the main term in \eqref{eq:regterm}, and it remains to estimate the errors resulting from the replacements in \eqref{eq:Mapprox}. 

Recall (see \eqref{eq:specdecH0mainres}) that $\mu_j$ and  $\bm u_j$ are the eigenvalues and the respective orthonormalized eigenvectors of $H_0$, i.e.
\begin{equation} \label{eq:specdecH0}
H_0 = \sum_j \mu_j \ket{\bm u_j} \bra{\bm u_j} \,. 
\end{equation}
Then, by means of~\eqref{eq:MDE}, spectral decomposition \eqref{eq:specdecH0} of $H_0$ and using Assumption \ref{ass:stateandobs}~(i) together with $[H_0, \Pi_\Delta] = 0$ and $\Pi_\Delta^2 = \Pi_\Delta$, we have that
\begin{equation} \label{eq:regtermspecdec}
\begin{split}
\text{lhs. of \eqref{eq:regterm}} = \Tr \big[\widetilde{\Theta}_1 A \widetilde{\Theta}_2 P\big] = \sum_{\mu_i,\mu_j \in I_\Delta} \langle \bm u_i, A \bm u_j \rangle \langle \bm u_j, P \bm u_i \rangle \widetilde{\vartheta}_{1}(i) \widetilde{\vartheta}_{2}(j) \,,
\end{split}
\end{equation}
where we denoted 
\begin{equation} \label{eq:specdecPsi}
\widetilde{\Theta}_1 := \sum_{\mu_j\in I_\Delta} \ket{\bm u_j} \bra{\bm u_j} \widetilde{\vartheta}_1(j) \quad \text{and} \quad \widetilde{\Theta}_2 := \sum_{\mu_j\in I_\Delta} \ket{\bm u_j} \bra{\bm u_j} \widetilde{\vartheta}_2(j) 
\end{equation}
with 
\begin{equation} \label{eq:thetadef}
\widetilde{\vartheta}_{1}(j) := \frac{1}{2 \pi \ii} \oint_{\gamma_1} \frac{\ee^{\ii t z_1}}{\mu_j - z_1 - \lambda^2 \llangle M_\lambda(z_1) \rrangle} \rd z_1\,, \quad \widetilde{\vartheta}_{2}(j) := \frac{1}{2 \pi \ii} \oint_{\gamma_2} \frac{\ee^{-\ii t z_2}}{\mu_j - z_2 - \lambda^2 \llangle M_\lambda(z_2) \rrangle} \rd z_2\,.
\end{equation}
Note that, by symmetry of the contours $\gamma_1$ and $\gamma_2$, we have that $\overline{\widetilde{\vartheta}_{1}(j)} = \widetilde{\vartheta}_{2}(j)$ and $\widetilde{\Theta}_1^* = \widetilde{\Theta}_2$. 

The key to approximating \eqref{eq:regtermspecdec} is the following lemma, whose proof is given at the end of the current Section \ref{subsubsec:regterm}. 
\begin{lemma}[First replacement lemma] \label{lem:replace}
Using the above notations and assumption, denote 
\begin{equation*} 
	\Theta_1 := \sum_{\mu_j \in I_\Delta} \ket{\bm u_j} \bra{\bm u_j} \vartheta_1(j) \quad \text{with} \quad \vartheta_{1}(j) := \frac{1}{2 \pi \ii} \oint_{\gamma_1} \frac{\ee^{\ii t z_1}}{\mu_j - z_1 - \lambda^2 \overline{m_0(E_0)}} \rd z_1
\end{equation*}
and $\Theta_2 := \Theta_1^*$ via $\vartheta_{2}(j) := \overline{\vartheta_{1}(j)}$, analogously to \eqref{eq:specdecPsi} and~\eqref{eq:thetadef}. Then it holds that
\begin{equation} \label{eq:replacereg}
\sup_{\mu_i \in I_\Delta} \left| \widetilde{\vartheta}_{1}(i) - \vartheta_1(i) \right| + \sup_{\mu_j \in I_\Delta} \left| \widetilde{\vartheta}_{2}(j) - \vartheta_2(j) \right| \lesssim \widetilde{\cE}_{\rm reg}
\end{equation}
for sufficiently small $\lambda > 0$ and $N$ large enough (dependent on $\lambda$, cf.~Lemma \ref{lemma:M_bounds}). Here, recalling \eqref{eq:rho0} for the definition of $\epsilon_0$, 
we denoted
\begin{equation} \label{eq:Ereg}
\widetilde{\cE}_{\rm reg} := \lambda^2 t \, \Delta + \lambda \,  (1 + \lambda^2 t)+ \frac{\lambda}{\Delta}\left(1 + \frac{\lambda}{\Delta}\right)+  \lambda^2 t \, \epsilon_0\,. 
\end{equation}
\end{lemma}
Therefore, by writing $\widetilde{\Theta} = \Theta + (\widetilde{\Theta} - \Theta)$ in \eqref{eq:regtermspecdec}, we find the lhs.~of \eqref{eq:regterm} to be given by
\begin{equation} \label{eq:regfourterms}
 \Tr \big[\Theta_1  A {\Theta}_2 P\big] + \Tr \big[(\widetilde{\Theta}_1- \Theta_1)  A \Theta_2 P\big] + \Tr \big[\Theta_1  A (\widetilde{\Theta}_2-\Theta_2) P\big] + \Tr \big[(\widetilde{\Theta}_1 - {\Theta}_1)  A (\widetilde{\Theta}_2-{\Theta}_2) P\big]\,. 
\end{equation}
The first term in \eqref{eq:regfourterms} precisely yields the result of \eqref{eq:regcompute} using Assumption \ref{ass:stateandobs}~(i). Using $\Vert A \Vert \lesssim 1$ and $\Tr[P] = 1$, the second and third
term in \eqref{eq:regfourterms} can be estimated by (a constant times)
\begin{equation*}
 \Vert \Theta_1 \Vert \,  \Vert \widetilde{\Theta}_2 - \Theta_2 \Vert +    \Vert \widetilde{\Theta}_1 - \Theta_1 \Vert \, \Vert \Theta_2 \Vert \lesssim \widetilde{\cE}_{\rm reg}\,. 
\end{equation*}
Here we used \eqref{eq:replacereg} and that  $\Vert \Theta_1 \Vert \le 1$ and $ \Vert \Theta_2 \Vert \le 1$
 as follows by the explicit expressions
\begin{equation*}
\Theta_1 = \ee^{\ii t(\Pi_\Delta H_0 \Pi_\Delta - \lambda^2 \overline{m_0(E_0)})} \quad \text{and} \quad \Theta_2 = \ee^{-\ii t(\Pi_\Delta H_0 \Pi_\Delta - \lambda^2 {m_0(E_0)})}\, 
\end{equation*}
and $\Im m_0(E_0)\ge 0$.
Similarly, applying \eqref{eq:replacereg} again, the fourth term in \eqref{eq:regfourterms} is bounded by $\mathcal{O}(\widetilde{\cE}_{\rm reg}^2)$. Collecting all four terms of \eqref{eq:Ereg}, this concludes the proof of \eqref{eq:regterm} with 
\begin{equation} \label{eq:Eregnew}
\cE_{\rm reg}:= \widetilde{\cE}_{\rm reg} + \widetilde{\cE}_{\rm reg}^2\,. 
\end{equation}

It remains to give the proof of Lemma \ref{lem:replace}.
\begin{proof}[Proof of Lemma \ref{lem:replace}] Since $\overline{\widetilde{\vartheta}_{1}(j)} = \widetilde{\vartheta}_{2}(j)$ and $\vartheta_{2}(j) := \overline{\vartheta_{1}(j)}$, we only estimate $\widetilde{\vartheta}_2(j)- \vartheta_2(j)$
 for arbitrary but fixed index $j$ such that  $\mu_j \in I_\Delta$. Moreover, for ease of notation, we completely drop the subscript $2$. 
	
As a first step, we split the contour into three parts: 
\begin{equation} \label{eq:contourdecomp}
\gamma  =\Gamma_1 \,  \dot{+} \, \Gamma_2 \,  \dot{+}\, \Gamma_3\,,
\end{equation}
where $\Gamma_1$ is the horizontal part of $\gamma$ with $\Re z \in I_{\Delc}$, $\Gamma_2$ is the horizontal part of $\gamma$ with $\Re z \notin I_{\Delc}$ and $\Gamma_3$ consists of the great arc of radius $R$ (cf.~Figure \ref{fig:contours}). We now estimate these three parts separately. 

For the first part, $\Gamma_1$, we have that (using the notation $m_\lambda(z) = \llangle M_{\lambda}(z) \rrangle$ from Lemma \ref{lemma:M_bounds})\footnote{Here and in the following, $|\rd z|$ denotes the total variation of the complex measure $\rd z$.}
\begin{equation} \label{eq:part1reg}
	\begin{split}
&\left| \int_{\Gamma_1}  \ee^{-\ii t z}  \left[\frac{1}{\mu_j - z - \lambda^2 m_{\lambda}(z)} - \frac{1}{\mu_j - z - \lambda^2
	 m_0(E_0)}\right] \rd z\right| \\
& \qquad \lesssim \int_{\Gamma_1} \frac{\lambda^2\,  \big(1/t + \lambda + \Delta + \epsilon_0\big)}{|\mu_j - z-\lambda^2 m_\lambda(z)| \, |\mu_j -z - \lambda^2 m_0(E_0)| } |\rd z| \lesssim \lambda^2t \,  \big(1/t + \lambda + \Delta + \epsilon_0\big)\,, 
	\end{split}
\end{equation}
uniformly in $\mu_j \in I_\Delta$. 
To go to the second line, we used that $|m_\lambda(z) - m_0(E_0)| \lesssim 1/t + \lambda + \Delta + \epsilon_0$. This follows by adding and subtracting $m_\lambda(E_0)$ and using $|m_\lambda(z) - m_\lambda(E_0)| \lesssim \Delta +1/t$ (using $|m_\lambda'(z)| \lesssim 1$ for $\Re z \in I_{2 \Delta}$; cf.~the last estimate in \eqref{eq:trM_local} from Lemma~\ref{lemma:M_bounds}) and $|m_\lambda(E_0) - m_0(E_0)| \lesssim \lambda + \epsilon_0$ (using that \eqref{eq:trM_est} holds down to the real line by combining it with \eqref{eq:Imm_lambdalower}).  For the final bound, we employed a Schwarz inequality for the integral and estimated the resulting integrals 
$$
  \int_{\Gamma_1} \frac{|\rd z|}{ |\mu_j - z-\lambda^2 m_\lambda(z)|^2} \lesssim (1 + \lambda^2) \, t \lesssim t\quad \text{and} \quad    \int_{\Gamma_1} \frac{|\rd z|}{ |\mu_j - z-\lambda^2 m_0(E_0)|^2} \lesssim t\,,
$$
by a change of variables $z \to z + \lambda^2 m_\lambda(z)$ using that $|m_\lambda'(z)| \lesssim 1$ for $z \in \Gamma_1$ by means of \eqref{eq:trM_local} together with $|\Im [z+\lambda^2 m_\lambda(z)] |\ge t^{-1}$, and $|m_0(E_0)| \lesssim 1$ together with 
$|\Im [z + \lambda^2m_0(E_0)] | \ge t^{-1}$,  respectively.

We now turn to the second part, i.e.~the integral similar to the left-hand side of ~\eqref{eq:part1reg}
but on the contour $\Gamma_2$. 
By means of $|m_0(E_0)| \lesssim 1$ and $|m_\lambda(z)| \le \lambda^{-1}$ (see the first estimate in \eqref{eq:trM_bounds}) we bound $|m_\lambda(z) - m_0(E_0)| \lesssim \lambda^{-1}$. Using $|m_0(E_0)| \lesssim 1$ and $|m_\lambda(z)| \le \lambda^{-1}$ again, together with $\mathrm{dist}(\mu_j, \Gamma_2) \gtrsim \Delta$, we find this second part to be bounded by (a constant times)
\begin{equation} \label{eq:part2reg}
\int_{\Gamma_2} \frac{\lambda}{|\mu_j - z|^2} |\rd z| \, \left(1 + \frac{\lambda}{\Delta}\right) \lesssim \frac{\lambda}{\Delta} \, \left(1 + \frac{\lambda}{\Delta}\right)\,, 
\end{equation}
again uniformly in $\mu_j \in I_\Delta$. 

Finally, we estimate the third part. By the exact same reasoning as for $\Gamma_2$, we arrive at the bound \eqref{eq:part2reg} with $\Delta$ replaced by $R$ and $\Gamma_2$ replaced by $\Gamma_3$. Hence, using that the radius $R$ of the semicircle is larger than one (see Figure \ref{fig:contours}),
 we find the third part to be bounded by $\mathcal{O}(\lambda/R)$,
 uniformly in $\mu_j \in I_\Delta$. 

Combining this with the error terms in \eqref{eq:part1reg} and~\eqref{eq:part2reg}, this concludes the proof. 
\end{proof}

\subsubsection{Proof of \eqref{eq:singterm}} \label{subsubsec:singterm}
Recall that $\tilsing$ denotes the singular term defined in \eqref{eq:singterm}. 
To carry out the analog of the approximation \eqref{eq:Mapprox}, we observe that the resolvent identity for $H_0$ implies
\begin{equation} \label{eq:M0_diff}
	\frac{z_{1,\lambda}-z_{2,\lambda}}{z_1-z_2}M_0(z_{1,\lambda})M_0(z_{2,\lambda})= \frac{M_0(z_{1,\lambda})-M_0(z_{2,\lambda})}{z_1-z_2} = \frac{1}{\pi}\int_\R \frac{\Im M_0(x_\lambda)}{(x-z_1)(x-z_2)}\mathrm{d}x,
\end{equation}
where we introduce the notation $z_{1,\lambda} := z_1+\lambda^2\overline{m_0(E_0)}$, $z_{2,\lambda} := z_2+\lambda^2m_0(E_0)$, and  $x_\lambda := x + \lambda^2m_0(E_0)$. 
Here, the second equality follows from the contour representation
of the resolvent $M_0$ of $H_0$, namely
\begin{equation}
	M_0(z) = \frac{1}{\pi}\int_\R \frac{\Im M_0(x+\ii\eta)}{x+\ii\eta-z}\mathrm{d}x, \quad \Im z >\eta>0.
\end{equation}
On the other hand, subtracting two instances of the MDE \eqref{eq:MDE} yields
\begin{equation} \label{eq:Mlambda_diff}
	\frac{M_\lambda(z_1)M_\lambda(z_2)}{1-\lambda^2\llangle M_\lambda(z_1)M_\lambda(z_2)\rrangle} = \frac{M_\lambda(z_1)-M_\lambda(z_2)}{z_1-z_2} = \frac{1}{\pi}\int_\R \frac{\Im M_\lambda(x)}{(x-z_1)(x-z_2)}\mathrm{d}x,
\end{equation}
where we used the matrix-valued analog of the Stieltjes representation for $M_\lambda(z)$ (cf.~\cite[Prop.~2.1]{AEK2020}),
\begin{equation}\label{eq:wStielt}
	M_\lambda(z)=\frac{1}{\pi}\int_\R\frac{\Im M_\lambda(x)}{x-z}\rd x, \quad z \in \C\setminus \R\,.
\end{equation}

In particular, identities \eqref{eq:M0_diff} and \eqref{eq:Mlambda_diff} suggest that the appropriate approximation for the factor $1-\lambda^2\llangle M_\lambda(z_1)M_\lambda(z_2)\rrangle$ in the denominator of \eqref{eq:singterm} is $(z_1-z_2)(z_{1,\lambda}-z_{2,\lambda})^{-1}$. Indeed, we prove that the following estimate holds.
\begin{lemma} \label{lemma:Fsing_rep}
	Under the Assumption \ref{ass:H0} and \ref{ass:stateandobs}, the singular term $\tilsing$ defined in \eqref{eq:singterm} satisfies
	\begin{equation} \label{eq:Fsing_replace}
		\bigl\lvert\tilsing - \sing\bigr\rvert \lesssim \mathcal{E}_{\mathrm{sing}} := (\Delta +\epsilon_0)(1+\lambda^2t)+ \lambda\bigl(1+\lambda^2t+\Delta^{-1}\log t\bigr)^2,
	\end{equation} 
	where the quantity $\sing$ is given by
	\begin{equation} \label{eq:Fsing_def}
		\begin{split}
			\sing := \frac{\lambda^2}{(2 \pi \ii)^2 }\oint_{\gamma_1}\oint_{\gamma_2}  &\ee^{\ii t (z_1 - z_2)} \frac{z_{1,\lambda}-z_{2,\lambda}}{z_1-z_2}\\&\times\llangle[\bigl]M_0(z_{1,\lambda})AM_0(_{2,\lambda})\rrangle \Tr \big[M_0(z_{1,\lambda})PM_0(z_{2,\lambda})\big] \rd z_1\rd z_2.
		\end{split}
	\end{equation}
\end{lemma}
We defer the proof of Lemma \ref{lemma:Fsing_rep} to the end of the current Section \ref{subsubsec:singterm}, and proceed to analyze the right-hand side of \eqref{eq:Fsing_def}.

Applying the identity \eqref{eq:M0_diff} to both traces in the integrand of \eqref{eq:Fsing_def}, we obtain the expression
\begin{equation} \label{eq:kernelproof}
	\sing = \int_\R\int_\R \llangle[\big]\Im M_0(x_\lambda)A \rrangle[\big]\Tr[\Im M_0(y_\lambda)P]F_{\lambda,t}(x,y)\rd x\rd y,
\end{equation}
where the function $F_{\lambda,t}(x,y)$ is defined as
\begin{equation} \label{eq:Flambdat}
	F_{\lambda,t}(x,y):=\frac{\lambda^2}{(2\pi\ii)^2\pi^2}\oint_{\gamma_1}\oint_{\gamma_2}\frac{ \ee^{\ii t (z_1 - z_2)} \, (z_1-z_2)}{(x-z_1)(x-z_2)(y-z_1)(y-z_2)(z_{1,\lambda}-z_{2,\lambda})}\rd z_2\rd z_1.
\end{equation}
Recall the contours $\gamma_1$ and $\gamma_2$ from Figure~\ref{fig:contours}, and that  $z_{1,\lambda}-z_{2,\lambda} = z_1-z_2-2\ii\alpha\lambda^2$.
Evaluating the contour integration over $\gamma_2$ in \eqref{eq:Flambdat} yields 
\begin{equation}\label{eq:residuesz2}
	\begin{split}
		F_{\lambda,t}(x,y) = \frac{\lambda^2}{\pi^2}\frac{1}{2\pi\ii}\oint_{\gamma_1}\biggl(&\frac{-\ee^{\ii t(z_1-x)}\chi(x)}{(x-y)(y-z_1)(x_\lambda-z_{1,\lambda})}
		+\frac{\ee^{\ii t(z_1-y)}\chi(y)}{(x-y)(x-z_1)(y_\lambda-z_{1,\lambda})}\\
		+&\frac{2\ii\alpha\lambda^2\ee^{-2\alpha\lambda^2t}}{(x-z_1)(y-z_1)(x_\lambda-z_{1,\lambda})(y_\lambda-z_{1,\lambda})}\biggr)\mathrm{d}z_1 - F_{\lambda,t}^{\mathrm{out}}(x,y),
	\end{split}
\end{equation}
where we define $\chi(z):=\mathds{1}_{\Omega_2}(z)$, and $\Omega_2$ is the compact connected component of $\C\backslash\gamma_2$, and the function $F_{\lambda,t}^{\mathrm{out}}(x,y)$ is defined as
\begin{equation}
	F_{\lambda,t}^{\mathrm{out}}(x,y) := \frac{\lambda^2}{\pi^2}\frac{1}{2\pi\ii}\oint_{\gamma_1}
	\frac{2\ii\alpha\lambda^2\ee^{-2\alpha\lambda^2t} \bigl(1-\chi(z_1-2\ii\alpha\lambda^2)\bigr)}{(x-z_1)(y-z_1)(x_\lambda-z_{1,\lambda})(y_\lambda-z_{1,\lambda})}\mathrm{d}z_1.
\end{equation}

We proceed to show that $F_{\lambda,t}^{\mathrm{out}}(x,y)$ contributes at most an $\mathcal{O}(\lambda^4(1+|\log\lambda|^2+ (\lambda^2t)^2))$ error to the right-hand side of \eqref{eq:kernelproof}.
Let $\gamma_{1,c}$ denote the set $\{z_1\in\gamma_1: z_1-2\ii\alpha\lambda^2 \notin \Omega_2\}$, then the contribution of $F_{\lambda,t}^{\mathrm{out}}(x,y)$ to the integral in \eqref{eq:kernelproof} is given by
\begin{equation}\label{eq:reserror}
\mathcal{E}_{\gamma_{1,c}}:=\frac{\alpha\lambda^4\ee^{-2\alpha\lambda^2t}}{\pi^3}\int_{\gamma_{1,c}}\int_\R\frac{\llangle[\big]\Im M_0(x_\lambda)A \rrangle[\big]\mathrm{d}x}{(x-z_1)(x_\lambda-z_{1,\lambda})}
\int_\R \frac{\Tr[\Im M_0(y_\lambda)P]\mathrm{d}y}{(y-z_1)(y_\lambda-z_{1,\lambda})}\rd z_1.
\end{equation}
Note that by choosing the radii of the contours $R\gtrsim 1$ large enough, we can assume that 
\begin{equation} \label{eq:dist_spec_gamma1c}
	\mathrm{dist}(\sigma(H_0),\gamma_{1,c})\gtrsim R.
\end{equation}
Using the spectral decomposition of $H_0$ from \eqref{eq:specdecH0}, the fact that $\int_\R \Im [(\mu_j-x_\lambda)]\mathrm{d}x = \pi$, and Assumption \ref{ass:stateandobs}, we conclude  that
\begin{equation} \label{eq:keyL1}
	\biggl\lvert\int_\R \llangle[\big]\Im M_0(x_\lambda)A\rrangle[\big]\rd x\biggr\rvert + \biggl\lvert\int_\R \Tr\bigl[\Im M_0(y_\lambda)P\bigr]\rd y\biggr\rvert \lesssim 1 + \lVert A\rVert \lesssim 1.
\end{equation}
Furthermore, the spectral decomposition for $H_0$ implies that for all $x$ with $\mathrm{dist}(x,\sigma(H_0)) \gtrsim 1$,
\begin{equation} \label{eq:keypointwise}
	\bigl\lvert\Tr\bigl[\Im M_0(x_\lambda)P\bigr]\bigr\rvert + \big|\llangle[\big]\Im M_0(x_\lambda)A \rrangle[\big]\big|\lesssim \lambda^2|x-E_0|^{-2}.
\end{equation}

Let $x_\pm$
denote the intersections of the contour $\gamma_1$ with $\R$ (see Figure \ref{fig:contours}), and define $\mathds{D}$ to be a union of two small disks around $x_\pm$,  $\mathds{D} := \{z\in \C : \min\{|z-x_-|,|z-x_+|\}\lesssim 1\}$.
Applying the Cauchy-Schwarz inequality to the $z_1$ integration in \eqref{eq:reserror}, and using the estimates \eqref{eq:dist_spec_gamma1c}-\eqref{eq:keypointwise} to bound the contribution coming from the outside of $\mathds{D}$, and applying the estimates \eqref{eq:dist_spec_gamma1c} and \eqref{eq:keypointwise} for all $x,y\in \R\cap\mathds{D}$, we obtain
\begin{equation} \label{eq:gamma1c_bound1}
	\bigl\lvert \mathcal{E}_{\gamma_{1,c}} \bigr\rvert \lesssim \lambda^4 \int_{\gamma_{1,c}\cap\mathds{D}} \biggl\lvert \frac{\lambda^2}{R^2}\int_{\R\cap\mathds{D}}\frac{\mathrm{d}x}{|x-z_1||x_\lambda-z_{1,\lambda}|} \biggr\rvert^2 |\mathrm{d}z_1| + \mathcal{O}\bigl( R^{-3}\lambda^4\bigr)\,.
\end{equation}

For $z_1$ on the horizontal linear segment of $\gamma_{1,c}\cap\mathds{D}$, we use that $\Im z_1 = -1/t$ to obtain
\begin{equation} \label{eq:line_bound}
	\frac{\lambda^2}{R^2}\int_{\R\cap\mathds{D}}\frac{\mathrm{d}x}{|x-z_1||x_\lambda-z_{1,\lambda}|} \lesssim \frac{1 + \lambda^2 t}{R^2},
\end{equation}
On the other hand, for $z_1$ lying on the circular arc parts of $\gamma_{1,c}\cap\mathds{D}$, we compute
\begin{equation} \label{eq:arc_bound}
	\frac{\lambda^2}{R^2}\int_{\R\cap\mathds{D}}\frac{\mathrm{d}x}{|x-z_1||x_\lambda-z_{1,\lambda}|} \lesssim \frac{\lambda^2}{R^2}\frac{
	|\log \lambda^2| + \bigl|\log |\eta_1|\bigr|+\bigl|\log|\eta_1 - 2\alpha\lambda^2|\bigr|}{|\eta_1| + \lambda^2}, 
\end{equation}
where $\eta_1:= \Im z_1$. Squaring the estimates \eqref{eq:line_bound} and \eqref{eq:arc_bound} and integrating them over the respective parts of $\gamma_{1,c}\cap\mathds{D}$, we conclude from \eqref{eq:gamma1c_bound1} that
\begin{equation} \label{eq:gamma_1cbound}
	\bigl\lvert \mathcal{E}_{\gamma_{1,c}} \bigr\rvert \lesssim R^{-3}\lambda^4 \bigl(1+ (\lambda^2t)^2+|\log\lambda|^2\bigr).
\end{equation}
Next, using residue calculus, we compute the first term on the right-hand side of \eqref{eq:residuesz2}, i.e., the contour integral over $\gamma_1$, to obtain
\begin{equation} \label{eq:F_approx}
	\begin{split}
		F_{\lambda,t}(x,y)=&~\frac{K_{\lambda,t}(x-y)}{\alpha\pi}-F_{\lambda,t}^{\mathrm{out}}(x,y) + \frac{K_{\lambda,t}(x-y)}{\alpha\pi}\bigl(\chi(x)\chi(y)-1\bigr)\\ &+ \frac{1}{2\pi^2}\frac{2\lambda^2\ee^{-2\alpha\lambda^2t}\bigl(\chi(x)-\chi(y)\bigr)^2}{|x-y|^2 + (2\alpha\lambda^2)^2}+ \frac{1}{\pi^2}\frac{2\ii\alpha\lambda^4\ee^{-2\alpha\lambda^2t}}{|x-y|^2+(2\alpha\lambda^2)^2}\frac{\chi(x)-\chi(y)}{x-y},
	\end{split}
\end{equation}
where $K_{\lambda,t}$ is the kernel defined in \eqref{eq:kernel}.
As we have proved above, the contribution of $F_{\lambda,t}^{\mathrm{out}}(x,y)$ to the integral in \eqref{eq:kernelproof}
admits the bound \eqref{eq:gamma_1cbound}. 
Similarly, using the estimates \eqref{eq:keyL1} and \eqref{eq:keypointwise}, it is straightforward to check that the third and fourth terms on the right-hand side of \eqref{eq:F_approx} contribute at most $\mathcal{O}(\lambda^4)$ to the right-hand side of \eqref{eq:kernelproof}, while the last term contributes at most $\mathcal{O}(\lambda^2)$. Therefore,
\begin{equation}
		\sing = \langle A \rangle_{\widetilde{P}_{\lambda,t}}\big(1+  \mathcal{O}(\epsilon_0 + \Delta + \lambda^2/\Delta)\big) + \mathcal{O}(\lambda^2),
\end{equation}
where we used the estimate $\pi\alpha = \pi^2\rho_0(E_0) = N^{-1}r\big(1+  \mathcal{O}(\epsilon_0 + \Delta + \lambda^2/\Delta)\big)$ that follows from Lemma \ref{lem:rho0approx} for $r$ defined in \eqref{eq:r_def}, 
and performed a change of variables $x \to x -  \lambda^2\Re m_0(E_0)$ and $y \to y - \lambda^2\Re m_0(E_0)$. We note that the $N^{-1}$ prefactor results from the different normalization of the trace in \eqref{eq:r_def} and \eqref{eq:rhoapprox}.
Furthermore, Proposition \ref{prop:kernel} below implies that under the Assumptions \ref{ass:H0} and \ref{ass:stateandobs}, $\lvert \langle A \rangle_{\widetilde{P}_{\lambda,t}} \rvert \lesssim 1$. Since the proof of Proposition \ref{prop:kernel} is independent of the statement of \eqref{eq:singterm}, this concludes the proof of \eqref{eq:singterm}. 

We proceed to prove Lemma \ref{lemma:Fsing_rep}.
\begin{proof}[Proof of Lemma \ref{lemma:Fsing_rep}]
Define the sequences of overlaps $\A_j := \langle \bm u_j, A\bm u_j\rangle$, and $\B_k := \langle \bm u_k, P\bm u_k\rangle$ where we recall from \eqref{eq:specdecH0} that $\bm u_j$'s are  the eigenvectors of $H_0$. 
Observe that the Assumption~\ref{ass:stateandobs} implies that 
\begin{equation} \label{eq:APbounds}
	\lVert \A\rVert_\infty \lesssim 1, \quad \B_k \ge 0 \quad \text{ and }\quad \lVert\B\rVert_1 = 1.
\end{equation}
Then, using the spectral decomposition  \eqref{eq:specdecH0}
of $H_0$  and the identity \eqref{eq:Mlambda_diff}, we rewrite $\tilsing$ in the following form,
\begin{equation}\label{eq:Fsing_expr}
	\tilsing = \sum_{j,k}\frac{\A_j\B_k}{N}\frac{1}{\pi}\int_\R \Im \widetilde{\nu}_j(x)\lambda^2\bigl\lvert \widetilde{\omega}_k(x)\bigr\rvert^2 \mathrm{d}x,
\end{equation}
where $\widetilde{\nu}_j(z) := (\mu_j-z-\lambda^2m_\lambda(z))^{-1}$ for $z \in \C$, and the functions $\widetilde{\omega}_k(x)$ are defined by the improper integrals\footnote{The integral in \eqref{eq:omega_tild_def} diverges logarithmically as $x$ approaches the intersection of the contour $\gamma_2$ with the real line. However, the contribution of such singularities to the $\mathrm{d}x$ integral on the right-hand side of \eqref{eq:Fsing_expr} is negligible. }
\begin{equation} \label{eq:omega_tild_def}
	\widetilde{\omega}_k(x) := \frac{1}{2\pi\ii}\oint\limits_{\gamma_2} \frac{\ee^{\ii t z}\widetilde{\nu}_k(z)}{x-z} \mathrm{d}z, \quad x \in \R.
\end{equation}
Here we adhere to the convention $m_\lambda(x) := \lim_{\eta\to +0} m_\lambda(x+\ii\eta)$.

The key estimates for proving \eqref{eq:Fsing_replace} are collected in the following Lemma that we prove at the end of the subsection.
\begin{lemma}[Second replacement lemma]  \label{lemma:replace2}
	Define the functions 
	\begin{equation}\label{eq:omega_def}
		\omega_k(x) := \frac{1}{2\pi\ii}\oint_{\gamma_2} \frac{\ee^{-\ii t z}}{x-z} \frac{1}{\mu_k - z_\lambda} \mathrm{d}z, \quad k\in[N],
	\end{equation}
	where we denote $z_\lambda := z + \lambda^2m_0(E_0)$. Then under the Assumptions \ref{ass:H0} and \ref{ass:stateandobs}, the estimates
	\begin{alignat}{2}
		\bigl\lvert \omega_k(x)-\widetilde{\omega}_k(x) \bigr\rvert &\lesssim \frac{(\Delta+\epsilon_0)(1+\lambda^2t)}{|\mu_k-x_\lambda|} + 1 + \lambda^2t+\Delta^{-1}\log t, \quad &&x\in I_{\Deld},\label{eq:omega_diff_I}\\
		\bigl\lvert \omega_k(x)\bigr\rvert  + \bigl\lvert\widetilde{\omega}_k(x) \bigr\rvert &\lesssim \frac{1}{|\mu_k-x_\lambda|} + 1 + \lambda^2t+\Delta^{-1}\log t, \quad &&x\in I_{\Deld},\label{eq:omega_sum_I}
	\end{alignat}
	\begin{equation} \label{eq:omega_sum_IC}
		\bigl\lvert \omega_k(x)\bigr\rvert  + \bigl\lvert\widetilde{\omega}_k(x) \bigr\rvert \lesssim  \Delta^{-1}\log t+R^{-1}, \quad x\in [-\tfrac{1}{2}R,\tfrac{1}{2}R]\backslash I_{\Deld},
	\end{equation}
 	hold for all $k$ with $\mu_k \in I_\Delta$.
\end{lemma}

Observe that applying the identities \eqref{eq:M0_diff}, \eqref{eq:Mlambda_diff}, and the spectral decomposition of $H_0$ to \eqref{eq:kernelproof} yields
\begin{equation} \label{eq:Fsing-Fsing}
	\tilsing - \sing = \mathcal{E}_{\mathrm{sing},1} + \mathcal{E}_{\mathrm{sing},2},
\end{equation}
where, recalling $x_\lambda := x-\lambda^2m_0(E_0)$, the quantities $\mathcal{E}_{\mathrm{sing},1}$ and $\mathcal{E}_{\mathrm{sing},2}$ are defined as
\begin{equation} \label{eq:sing_err_1}
	\mathcal{E}_{\mathrm{sing},1}:= \sum_{j,k}\frac{\A_j\B_k}{N}\frac{1}{\pi}\int_\R\Im\bigl[\widetilde{\nu}_j(x)- (\mu_j-x_\lambda)^{-1}\bigr] \lambda^2\bigl\lvert \widetilde{\omega}_k(x)\bigr\rvert^2\mathrm{d}x.
\end{equation}
\begin{equation} \label{eq:sing_err_2}
	\mathcal{E}_{\mathrm{sing},2} := \sum_{j,k}\frac{\A_j\B_k}{N}\frac{1}{\pi}\int_\R\Im\bigl[(\mu_j-x_\lambda)^{-1}\bigr]\lambda^2 \biggl(\bigl\lvert\widetilde{\omega}_k(x)\bigr\rvert^2- \bigl\lvert\omega_k(x)\bigr\rvert^2\biggr)\mathrm{d}x.
\end{equation}

First, we estimate the quantity $\mathcal{E}_{\mathrm{sing},1}$ defined in \eqref{eq:sing_err_1}.
In the regime $x\in I_{\Deld}$, the bounds in \eqref{eq:trM_local} imply that
\begin{equation} \label{eq:nu_diff_I}
	\bigl\lvert \widetilde{\nu}_j(x)-(\mu_j-x_\lambda)^{-1}\bigr\rvert \lesssim (\Delta+\epsilon_0) \frac{\lambda^2}{|\mu_j-x_\lambda|^2}, \quad x \in I_{\Deld}.
\end{equation}
On the other hand, it is straightforward to check that the regime $|x| \ge \frac{1}{2}R$ contributes at most $\mathcal{O}(\lambda^2 R^{-2})$ to the integral on the right-hand side of \eqref{eq:sing_err_1}. We note that the logarithmic singularity resulting from the contour $\gamma_2$ intersecting the real line is removed by the $x$ integration.

Therefore, estimates \eqref{eq:APbounds}, \eqref{eq:omega_sum_I}, \eqref{eq:omega_sum_IC}, and \eqref{eq:nu_diff_I} imply that
\begin{equation} \label{eq:Eps_sing1}
	\begin{split}
		\bigl\lvert \mathcal{E}_{\mathrm{sing},1} \bigr\rvert \lesssim&~ \sum_{j,k}\frac{|\A_j|\B_k}{N}\int_{I_{3\Delta}}\frac{\lambda^2(\Delta+\epsilon_0)}{|\mu_j-x_\lambda|^2} \frac{\lambda^2}{|\mu_k-x_\lambda|^2}\mathrm{d}x\\
		&+\lambda^2(1 + \lambda^2t+\Delta^{-1}\log t)^2\sum_{j}\frac{1}{N}\int_{\R}\Im\bigl[\widetilde{\nu}_j(x) + (\mu_j-x_\lambda)^{-1}\bigr] \mathrm{d}x\\
		\lesssim&~\sum_{j,k}\frac{|\A_j|\B_k}{N}\int_{\R} \frac{\lambda^2(\Delta+\epsilon_0)}{|\mu_j-x_\lambda|^2} \frac{\lambda^2}{|\mu_k-x_\lambda|^2}\mathrm{d}x + \lambda^2\bigl(1+\lambda^2 t + \Delta^{-1}\log t\bigr)^2,
	\end{split}
\end{equation}
where in the fist step we used the bound $|\Im[\widetilde{\nu}_j(x) - (\mu_j-x_\lambda)^{-1}]| \le \Im[\widetilde{\nu}_j(x) + (\mu_j-x_\lambda)^{-1}]$ to estimate the contribution coming form the second term on the right-hand side of \eqref{eq:omega_sum_I} and the regime $x \in [-\frac{1}{2}R, \frac{1}{2}R] \backslash I_{3\Delta}$; and in the second step we used that $\int_\R\Im[(\mu_j-x_\lambda)^{-1}] \mathrm{d}x = \pi$ and $\int_\R \frac{1}{N}\sum_j \Im\widetilde{\nu}_j(x)\mathrm{d}x = \int_R \Im m_\lambda(x)\mathrm{d}x = \pi$ (see, e.g., Proposition 2.1 and Eq.~(2.9) in \cite{AEK2020}).
Computing the integral in the second term on the right-hand side of \eqref{eq:Eps_sing1} explicitly, and using the spectral decomposition of $H_0$, we deduce from the admissibility of $E_0$ that
\begin{equation} \label{eq:sumCauchy_bound}
	\sum_{j,k}\frac{|\A_j|\B_k}{N}\int\limits_\R\frac{\lambda^2(\Delta+\epsilon_0)}{|\mu_j - x_\lambda|^2}  \frac{\lambda^2\mathrm{d}x}{|\mu_k-x_\lambda|^2} \lesssim (\Delta+\epsilon_0) \sup_{\mu_k\in I_{\Delta}} \llangle[\bigl] \Im M_0(\mu_k + 2\ii \lambda^2 \alpha) \rrangle[\bigr] \lesssim \Delta+\epsilon_0.
\end{equation}
Here we employed \eqref{eq:APbounds} and the estimate $\langle \A, X \B\rangle \lesssim \lVert \A \rVert_\infty \lVert \B \rVert_1 \sup\limits_{k\in\supp(\B)} \sum_j |X_{jk}|$. Hence, we conclude that 
\begin{equation} \label{eq:Eps_sing1_bound}
	\bigl\lvert \mathcal{E}_{\mathrm{sing},1} \bigr\rvert \lesssim \Delta +\epsilon_0+ \lambda^2\bigl(1+\lambda^2t+\Delta^{-1}\log t\bigr)^2.
\end{equation}

We proceed to estimate the quantity $\mathcal{E}_{\mathrm{sing},2}$ defined in \eqref{eq:sing_err_2}. 
We note again that the contribution of the regime $|x| \ge \frac{1}{2}R$ to the integral on the right-hand side of \eqref{eq:sing_err_2} is bounded by $\mathcal{O}(\lambda^2R^{-2})$. Therefore, combining the estimates \eqref{eq:omega_diff_I}, \eqref{eq:omega_sum_I}, and \eqref{eq:omega_sum_IC} yields the bound
\begin{equation} \label{eq:Eps_sing2_bound}
	\begin{split}
	\bigl\lvert \mathcal{E}_{\mathrm{sing},2}\bigr\rvert \lesssim 
	(\Delta+\epsilon_0)(1+\lambda^2t)
	+\lambda \bigl(1+\lambda^2t+\Delta^{-1}\log t\bigr)^2,
	\end{split}
\end{equation}
obtained similarly to \eqref{eq:Eps_sing1} and \eqref{eq:Eps_sing1_bound}.
Together with \eqref{eq:Fsing-Fsing}, the bound \eqref{eq:Eps_sing1_bound} and \eqref{eq:Eps_sing2_bound} conclude the proof of \eqref{eq:Fsing_replace}.
\end{proof}

It remains to prove Lemma \ref{lemma:replace2}.
\begin{proof}[Proof of Lemma \ref{lemma:replace2}]
	Throughout the proof we assume that $k\in [N]$ satisfies $\mu_k \in I_{\Delta}$, and $x\in \R$ satisfies $|x| \le \frac{1}{2}R$.   
	We introduce the auxiliary quantities 
	\begin{equation}
		\check{\nu}_k(z) := \frac{1}{\mu_k - z -\lambda^2m_\lambda(\mu_k)}, \quad 
		\check{\omega}_k(x) := \frac{1}{2\pi\ii}\oint_{\gamma_2}\frac{\ee^{-\ii tz}}{x-z}\check{\nu}_k(z)\mathrm{d}z.
	\end{equation}
	An explicit computation using the residue calculus reveals that 
	\begin{equation} \label{eq:omega_expr}
		\check{\omega}_k(x)  = \frac{\ee^{-\ii t x } - \ee^{-\ii t (\mu_k - \lambda^2m_\lambda(\mu_k))}}{\mu_k -x - \lambda^2m_\lambda(\mu_k)}, \quad \omega_k(x)  = \frac{\ee^{-\ii t x } - \ee^{-\ii t (\mu_k - \lambda^2 m_0(E_0))}}{\mu_k -x_\lambda}.
	\end{equation}	
	Furthermore, using the bound in \eqref{eq:trM_local}, we obtain
	\begin{equation} \label{eq:omega-check_bound}
		\bigl\lvert \omega_k(x) - \check{\omega}_k(x) \bigr\rvert \lesssim \biggl(\frac{1}{|\mu_k-x_\lambda|^2} + \frac{t}{|\mu_k-x_\lambda|}\biggr)\lambda^2|\mu_k - E_0| \lesssim \frac{\Delta (1+\lambda^2 t)}{|\mu_k-x_\lambda|},
	\end{equation}
	where we additionally applied the estimate
	\begin{equation} \label{eq:abs_sim}
		\bigl|y+\mathcal{O}(\eta)\bigr| + \eta \sim |y| + \eta, \quad y\in\R, \eta>0.
	\end{equation}
	
	We decompose the contour $\gamma_2 = \Gamma_1\,\dot{+}\,\Gamma_2 \,\dot{+}\,\Gamma_3$ according to \eqref{eq:contourdecomp}.
	It is straightforward to check that for $\nu^\#(z)$ denoting one of $\widetilde{\nu}_k(z)$, $\check{\nu}_k(z)$ or $(\mu_k - z_\lambda)^{-1}$,
	\begin{equation} \label{eq:Gamma2+3_tail}
		\biggl\lvert \int_{\Gamma_2\dot{+}\Gamma_3} \frac{\ee^{-\ii tz}}{x-z}\nu^\#_k(z)\mathrm{d}z \biggr\rvert \lesssim \frac{\log t}{\Delta} + \frac{1}{R},
	\end{equation}
	where we used that $\Im z = t^{-1}$ for all $z\in\Gamma_1$.
	Therefore, rewriting the left-hand sides of \eqref{eq:omega_diff_I}-\eqref{eq:omega_sum_IC} using the integral definitions \eqref{eq:omega_tild_def} and\eqref{eq:omega_def}, it suffices to estimate the contributions coming from the segment $\Gamma_1\subset \gamma_2$.

	Using \eqref{eq:abs_sim} and \eqref{eq:trM_local}, we deduce that for all $z\in \Gamma_1$, defined in \eqref{eq:contourdecomp}, 
	\begin{equation} \label{eq:nu_check_bound_I}
		\bigl\lvert \check{\nu}_k(z) - \widetilde{\nu}_k(z) \bigr\rvert \lesssim \frac{\lambda^2}{|\mu_k-z_\lambda|}.
	\end{equation}
	Integrating the bound \eqref{eq:nu_check_bound_I} then yields
	\begin{equation}
		\biggl\lvert \int_{\Gamma_1} \frac{\ee^{-\ii tz}}{x-z}\bigl[\widetilde{\nu}_k(z)-\check{\nu}_k(z)\bigr]\mathrm{d}z \biggr\rvert \lesssim \bigl( 1 + \lambda^2t \bigr)\mathds{1}_{x\in I_{\Deld}} + \bigl(\Delta^{-1}\log t\bigr) \mathds{1}_{x\notin I_{\Deld}},
	\end{equation}
	which, together with \eqref{eq:omega-check_bound}, \eqref{eq:Gamma2+3_tail} immediately implies \eqref{eq:omega_diff_I} after writing $\widetilde{\omega}_k - \omega_k = (\widetilde{\omega}_k-\check{\omega}_k) + (\check{\omega}_k-\omega_k) $.
	
	On the other hand, noting that $|\omega_k(x)|+ |\check{\omega}_k(x)| \lesssim |\mu_k-x_\lambda|^{-1}$ by \eqref{eq:omega_expr}, and combining the estimates \eqref{eq:omega-check_bound} and \eqref{eq:Gamma2+3_tail} yields \eqref{eq:omega_sum_I} and \eqref{eq:omega_sum_IC}. This concludes the proof of Lemma \ref{lemma:replace2}.
\end{proof}

\subsection{Step (iii): Limiting behavior of the singular term and proof of Theorem \ref{thm:main}~(b)} \label{subsec:step3} 
We organize the result of approximating $\langle A \rangle_{\widetilde{P}_{\lambda,t}}$ in the following proposition. 
\begin{proposition} \label{prop:kernel}
	Under the assumptions of Theorem~\ref{thm:main}, and with $\widetilde{P}_\lambda$ defined as in \eqref{eq:kernelstateapprox}, we have that, for any fixed $T \in (0,\infty)$ and recalling $\alpha = \pi \rho_0(E_0)$
	\begin{equation*}
		\limsup_{\Delta\rightarrow0}\limsup_{\substack{t\rightarrow\infty\\\lambda\rightarrow0\\\lambda^2 t=T}}\limsup_{N\rightarrow\infty} \bigl\lvert\langle A \rangle_{\widetilde{P}_{\lambda,t}} - (1- \ee^{-2\alpha\lambda^2t})\langle A \rangle_{\widetilde{P}_\lambda} \bigr\rvert \lesssim T\ee^{-2\alpha T}\,.
	\end{equation*}
\end{proposition}

Given Proposition \ref{prop:kernel}, Theorem \ref{thm:main}~(b) immediately follows. \qed
\begin{proof}[Proof of Proposition~\ref{prop:kernel}]
	First, we observe that representing $\Im M_\lambda$ in spectral decomposition of $H_0$, the quantity $\langle A \rangle_{\widetilde{P}_{\lambda,t}}$ with $P_{\lambda, t}$ defined in \eqref{eq:kernelstate}, can be rewritten in the from 
	\begin{equation} \label{eq:APtildet}
		\langle A \rangle_{\widetilde{P}_\lambda,t} = \frac{1}{r}\sum_{j,k}\A_j\B_k\int_{\R} \phi_{\alpha\lambda^2}(x-\mu_j) \bigl(K_{\lambda,t} * \phi_{\alpha\lambda^2}\bigr)(x-\mu_k) \mathrm{d}x,
	\end{equation}
	where $r=\int_\R \Tr[\Im M_0(x + \ii \alpha\lambda^2)]\langle\Im M_0(x + \ii \alpha \lambda^2)\rangle_P\mathrm{d}x > 0$ has already been introduced in Remark \ref{rmk:relax}~(ii), and we denoted $\phi_{\eta} := \Im[(x-\ii\eta)^{-1}]$. Recall that $\mu_j, \bm u_j$ are the eigenvalues and the respective eigenvectors of $H_0$, and $\A_j := \langle \bm u_j, A\bm u_j\rangle$, $\B_j := \langle \bm u_j, P\bm u_j\rangle$.
	Applying the Parseval-Plancherel identity to the right-hand side of \eqref{eq:APtildet} yields
	\begin{equation} \label{eq:Shat}
		\langle A \rangle_{\widetilde{P}_{\lambda,t}} = \frac{1}{r}\sum_{j,k}\A_j\B_k \Phi_{\lambda,t}(\mu_j-\mu_k), \quad \Phi_{\lambda,t}(u) := 
		\frac{\pi^2}{(2\pi)^{1/2}}\int_\R \ee^{-2\alpha\lambda^2|q|-\ii u q}\, \widehat{K_{\lambda,t}}(q)\mathrm{d}q,
	\end{equation}
	where we used the fact that $\widehat{\cau_\eta}(q) = (\frac{\pi}{2})^{1/2}\ee^{-\eta|q|},\, \eta>0$ (recall Footnote \ref{ftn:FT}).
	
	A direct computation starting with \eqref{eq:kernel} reveals that 
	\begin{equation} \label{eq:kernelHat}
		\widehat{K_{\lambda,t}}(q) = \begin{cases}
			(2\pi)^{-1/2}\bigl(1 - \ee^{-2\alpha\lambda^2 (t-|q|)} \bigr) \quad &\text{for} \quad |p| \le t\,, \\
			0 \quad &\text{for} \quad |p| > t\,,
		\end{cases}
	\end{equation}
	and implies that $\Phi_{\lambda,t}(u)$ admits the explicit expression 
	\begin{equation} \label{eq:Phi_fomula}
			\Phi_{\lambda,t}(u) = \bigl(1-\ee^{-2\alpha\lambda^2 t}\bigr) \pi\cau_{2\alpha\lambda^2}(u)+ \mathfrak{R}_{\lambda,t}(u),
	\end{equation}
	where the function $\mathfrak{R}_{\lambda, t}(u)$ is defined by
	\begin{equation} \label{eq:Phitilde}
		\mathfrak{R}_{\lambda,t}(u) := \pi\ee^{-2\alpha\lambda^2 t}\phi_{2\alpha\lambda^2}(u)\biggl(1-\cos(tu)-2\alpha\lambda^2 t \frac{\sin(tu)}{tu}\biggr).
	\end{equation}
	Observe that the contribution of the first term on the right-hand side of \eqref{eq:Phi_fomula} to $\langle A \rangle_{\widetilde{P}_{\lambda,t}}$ is given by $(1-\ee^{-2\lambda^2\alpha t})\langle A \rangle_{\widetilde{P}_\lambda}$, since
	\begin{equation}
		\langle A \rangle_{\widetilde{P}_\lambda} =
		\frac{1}{r}\sum_{j,k}\A_j\B_k\int_\R \cau_{\alpha\lambda^2}(x-\mu_j) \cau_{\alpha\lambda^2}(x-\mu_k)\mathrm{d}x.
	\end{equation}
	Here we used the definition of the state $\widetilde{P}_\lambda$ in \eqref{eq:kernelstateapprox}, and the Parseval-Plancherel identity.
	
	The key observation is that the contribution of the remaining $\mathfrak{R}_{\lambda,t}(\mu_j-\mu_k)$ term 
	\begin{equation} \label{eq:Rdef}
		\mathcal{R} := \sum_{j,k}\frac{\A_j\B_k}{r} \mathfrak{R}_{\lambda,t}(\mu_j-\mu_k)
	\end{equation}
	in \eqref{eq:Phi_fomula} to $\langle A \rangle_{\widetilde{P}_{\lambda,t}}$ admits the bound\footnote{Inequality \eqref{eq:youngPhi} can be interpreted as a discrete analog of Young's convolution inequality, which can not be evoked directly since the eigenvalues $\mu_j$ do not form a group under addition.}
	\begin{equation} \label{eq:youngPhi}
		\big\lvert\mathcal{R}\big\rvert \le \sup_{k\in \supp(\B)}\frac{1}{r}\sum_{j} \bigl|\mathfrak{R}_{\lambda,t}(\mu_j-\mu_k) \bigr|\cdot\lVert\A\rVert_\infty \lVert\B\rVert_1\,.
	\end{equation}
	
	Observe that there exists a constant $C>0$ such that for any $\xi >0$ and $t>0$, we have that $\cau_{ \xi}(u) (1-\cos(tu)) \le C\xi t \cau_{1/t}(u)$ for all $u \in \R$. This follows immediately from the fact that the function $s \mapsto (s^2+1)(1-\cos s)/s^2$ is uniformly bounded on $\R$. Therefore, the function $\mathfrak{R}$ admits the bound
	\begin{equation} \label{eq:Phi_bound}
		|\mathfrak{R}_{\lambda,t}(u)| \le 2\pi\alpha\lambda^2 t\,\ee^{-2\alpha\lambda^2 t}\biggl(C\phi_{1/t}(u)+\phi_{2\alpha\lambda^2}(u)\biggr), \quad u\in\R.
	\end{equation}
	Summing the bound \eqref{eq:Phi_bound} over $u = \mu_j$ yields
	\begin{equation}
		\frac{1}{r}\sum_{j} \bigl|\mathfrak{R}_{\lambda,t}(\mu_j-\mu_k) \bigr| \lesssim \lambda^2 t \,\ee^{-2\alpha\lambda^2 t}\frac{N}{r}\Im\llangle M_0(\mu_k+2\ii\alpha\lambda^2)+M_0(\mu_k+\ii/t)\rrangle. 
	\end{equation}
	Using the localization of the state $P$ as in \eqref{eq:Plocalize}, the admissibility of $E_0$ in \eqref{eq:E0}, and the first line of \eqref{eq:rhoapprox} to deduce that $r \sim N(1+ \mathcal{O}(\epsilon_0 + \Delta + \lambda^2/\Delta))$, we obtain
	\begin{equation} \label{eq:Rsum_est}
		\sup_{k\in \supp(\B)}\frac{1}{r}\sum_{j} \bigl|\mathfrak{R}_{\lambda,t}(\mu_j-\mu_k) \bigr| \lesssim \lambda^2 t\,\ee^{-2\alpha\lambda^2 t}\bigl(1 + \epsilon_0\bigr)\big(1+\epsilon_0 + \Delta + \lambda^2/\Delta\big).
	\end{equation}
	This concludes the proof of Proposition \ref{prop:kernel}.	
\end{proof}

\subsection{Relaxation formula: Proof of Corollary \ref{cor:relax}} \label{subsec:relaxproof} 
Estimates \eqref{eq:short} and \eqref{eq:long} in items  (a) and  (b), respectively, follow immediately from Theorem \ref{thm:main}.

To prove \eqref{eq:LORRsmall} in item (c), observe that plugging the estimate \eqref{eq:LORP} from the Definition \ref{def:LOR} of local overlap regularity into \eqref{eq:Rdef} yields
\begin{equation} \label{eq:LOR_Rsums}
	\begin{split}
		\bigl\lvert\mathcal{R}\bigr\rvert \lesssim&~ |\mathfrak{A}| \sup\limits_{k\in\supp(\B)} \biggl\lvert\frac{1}{N}\sum_{\mu_j\in I_{2\Delta}} \mathfrak{R}_{\lambda,t}(\mu_j-\mu_k)\biggr\rvert 
		+ \sup\limits_{k\in\supp(\B)} \biggl\lvert\sum_{\mu_j\in I_{2\Delta}} \frac{\A_j-\mathfrak{A}}{N} \mathfrak{R}_{\lambda,t}(\mu_j-\mu_k)\biggr\rvert\\ 
		&+ \sup\limits_{k\in\supp(\B)} \biggl\lvert \frac{1}{N}\sum_{\mu_j\notin I_{2\Delta}} \A_j \mathfrak{R}_{\lambda,t}(\mu_j-\mu_k) \biggr\rvert,
	\end{split}
\end{equation}
where we used that  $|r| \sim N\, (1+ \mathcal{O}(\epsilon_0 + \Delta + \lambda^2/\Delta))$ by the first line of \eqref{eq:rhoapprox} from Lemma \ref{lem:rho0approx}.

Applying the estimates analogous to \eqref{eq:youngPhi} and \eqref{eq:Rsum_est} to the second sum on the right-hand side of \eqref{eq:LOR_Rsums}, we deduce the bound
\begin{equation} \label{eq:cE_LOR_error}
	\sup\limits_{k\in\supp(\B)} \biggl\lvert\sum_{\mu_j\in I_{2\Delta}} \frac{\A_j-\mathfrak{A}}{N} \mathfrak{R}_{\lambda,t}(\mu_j-\mu_k)\biggr\rvert \lesssim \big| \cE_{{\rm LOR}} \big| \cdot\big(1+\lambda^2/\Delta\big)
\end{equation}

Note that by \eqref{eq:APbounds} and the uniform bound 
$$
|\mathfrak{R}_{\lambda,t}(u)| \lesssim \frac{\lambda^2}{\Delta^2}, \quad\mbox{for} \quad  |u| \gtrsim \Delta
$$
following from~\eqref{eq:Phi_bound}, the tail sum, i.e., the second line of \eqref{eq:LOR_Rsums}, admits the estimate
\begin{equation} \label{eq:LOR_tail_error}
	\sup\limits_{k\in\supp(\B)} \biggl\lvert \frac{1}{N}\sum_{\mu_j\notin I_{2\Delta}} \A_j \mathfrak{R}_{\lambda,t}(\mu_j-\mu_k) \biggr\rvert \lesssim \frac{\lambda^2}{ \Delta^2}.
\end{equation}

Therefore, it remains to estimate the first term on the right-hand side of \eqref{eq:LOR_Rsums}.
Since the function $\mathfrak{R}_{\lambda,t}(u)$ is holomorphic in $u$ for $|\Im u| \le \alpha\lambda^2$, we obtain the following series of estimates,
\begin{equation} \label{eq:Riemann_sum}
	\begin{split}
			\frac{1}{N}\sum_{\mu_j\in I_{2\Delta}} \mathfrak{R}_{\lambda,t}(\mu_j-\mu_k) =&~ \frac{1}{2\pi\ii} \oint_{\gamma} \mathfrak{R}_{\lambda,t}(z-\mu_k)\llangle M_0(z)\rrangle\mathrm{d}z + \mathcal{O}\biggl(\frac{\lambda^2 }{\Delta^2}\biggr)\\
			=&~\frac{1}{2\pi\ii} \oint_{\gamma} \mathfrak{R}_{\lambda,t}(z-\mu_k)m_0(z)\mathrm{d}z + \mathcal{O}\biggl(\eta_0+\frac{\epsilon_0}{\lambda^2}+\frac{\lambda^2 }{\Delta^2}\biggr)\\
			=&~\int_{I_{2\Delta}} \mathfrak{R}_{\lambda,t}(u-\mu_k)\rho_0(u)\mathrm{d}u + \mathcal{O}\biggl(\eta_0+\frac{\epsilon_0}{\lambda^2}+\frac{\lambda^2 }{\Delta^2}\biggr),
	\end{split}
\end{equation}
where the contour $\gamma$ is defined to be a rectangle of height $2\eta_0$ and width $4C$ centered at $E_0$, and the constant $C\sim 1$ is chosen in such a way that $\sigma(H_0) \subset [E_0-C,E_0+C]$. Here, in the first step, we used residue calculus and an estimate analogous to \eqref{eq:LOR_tail_error} to extend the sum to all $\mu_j$'s. The second step follows by integrating the estimate \eqref{eq:rho0} on the horizontal segments of $\gamma$ and bounding the contribution of the vertical segments of the contour $\gamma$ by $\mathcal{O}(\eta_0)$. Finally, the third step is a consequence of the Stieltjes representation \eqref{eq:m0} and $|\mathfrak{R}_{\lambda,t}(u)| \lesssim \Delta^{-2}\lambda^2$ for $|u| \gtrsim \Delta$. Using the estimate $\rho_0(u) = \rho_0(E_0) + \mathcal{O}(\Delta)$ for all $u\in I_{2\Delta}$ by admissibility of $E_0$ as in  \eqref{eq:admiss_spec}, we conclude that 
\begin{equation}
	\int_{I_{2\Delta}} \mathfrak{R}_{\lambda,t}(u-\mu_k)\rho_0(u)\mathrm{d}u =\rho_0(E_0) \int_{I_{2\Delta}} \mathfrak{R}_{\lambda,t}(u-\mu_k)\mathrm{d}u  + \mathcal{O}(\Delta)\,,
\end{equation}
where we used $\int_\R |\mathfrak{R}_{\lambda, t}(u)| \rd u \lesssim \lambda^2 t\,\ee^{-2\alpha\lambda^2 t} \lesssim 1$ as a consequence of \eqref{eq:Phi_bound}.
Moreover, a direct computation starting with \eqref{eq:Phitilde} reveals that 
\begin{equation} \label{eq:R_integrals}
	\int_\R \mathfrak{R}_{\lambda,t}(u-\mu_k)\mathrm{d}u = 0 \qquad 
	\text{and} \qquad \int_{\R\backslash I_{2\Delta}} \bigl\lvert\mathfrak{R}_{\lambda,t}(u-\mu_k)\bigr\rvert\mathrm{d}u \lesssim \frac{\lambda^2}{\Delta} \,.
\end{equation}
Hence, combining estimates \eqref{eq:LOR_Rsums}--\eqref{eq:R_integrals} yields
\begin{equation}
	\bigl\lvert\mathcal{R}\bigr\rvert \lesssim \Delta + \eta_0+\lambda^{-2}\epsilon_0+\Delta^{-2}\lambda^2 \,,
\end{equation}
which implies the $\mathcal{R}$-part of \eqref{eq:LORRsmall}; the $\mathcal{E}$-part is an immediate consequence of Theorem \ref{thm:main}~(a).

To complete the proof under the weaker assumption on $\langle \bm u_j, A\bm u_j\rangle$, stated in Footnote \ref{ftn:LOR}, we first \emph{uniformly} approximate $\mathfrak{A}_N$ by a real analytic function $\mathfrak{A}_{N, \ell} : I_{2 \Delta} \to \R$ with $\ell = \ell(\lambda, t)> 0$ (to be chosen below), which can be analytically extended to $\{z \in \C : \mathrm{dist}(z, I_{2 \Delta}) < \ell\}$ and satisfy $\sup_{N \in \N}\Vert \mathfrak{A}_{N, \ell} - \mathfrak{A}_N \Vert_{\infty} \to 0$ as $\ell \to 0$. Such $\mathfrak{A}_{N, \ell}$ can be explicitly constructed, e.g., by convolution of $\mathfrak{A}_N$ with a Gaussian having variance of order $\ell$. For ease of notation, we shall now drop the subscript $N$. Then, the error term $\mathfrak{A} - \mathfrak{A}_\ell$ is easily seen to give a vanishing contribution (as $\ell \to 0$) by means of \eqref{eq:Rsum_est}. Next, observe that using analyticity of $\mathfrak{A}_\ell$ and reasoning as in the proof for the case of $\mathfrak{A}$ being constant above, we obtain
\begin{equation} \label{eq:LOR_analytic}
	\bigl\lvert\mathcal{R}\bigr\rvert \lesssim \sup_{\mu_k\in I_\Delta}\biggl\lvert \int_{I_{2\Delta}} \bigl(\mathfrak{A}_\ell(u)-\mathfrak{A}_\ell(\mu_k)\bigr)\mathfrak{R}_{\lambda,t}(u-\mu_k)\mathrm{d}u \biggr\rvert + \Delta + \eta_0+\lambda^{-2}\epsilon_0+\Delta^{ -2}\lambda^2.
\end{equation}
Since $\mathfrak{A}_\ell(z)$ is analytic in the strip of width $\ell$, \eqref{eq:Phi_bound} implies that the integral on the right-hand side of \eqref{eq:LOR_analytic} is bounded by 
\begin{equation}
	\frac{1}{\ell}\int_{I_{2\Delta}} \biggl(\frac{\lambda^2 |u-\mu_k|}{|u-\mu_k|^2+(2\alpha\lambda^2)^2} + \frac{t^{-1} |u-\mu_k|}{|u-\mu_k|^2+t^{-2}}\biggr)\mathrm{d}u \lesssim \frac{\lambda^2 |\log \lambda| + t^{-1} \log t}{\ell},
\end{equation}
uniformly in $k$ such that $\mu_k \in I_\Delta$. Hence, choosing, say, $\ell := \lambda + t^{-1/2}$, this concludes the proof of Corollary \ref{cor:relax}. \qed

\subsection{Microcanonical average: Proof of Theorem \ref{thm:PreT2}} \label{subsec:proofmc}
Using \eqref{eq:kernelstateapprox}, we start by writing out
\begin{equation} \label{eq:mcproofstart}
\langle A \rangle_{\widetilde{P}_\lambda} = \frac{\int_\R  \llangle \Im M_0(x + \ii \alpha \lambda^2) A \rrangle \,  \langle \Im M_0(x + \ii \alpha \lambda^2) \rangle_P  \, \rd x }{\int_\R \llangle \Im M_0(x + \ii \alpha \lambda^2) \rrangle \,  \langle \Im M_0(x + \ii \alpha \lambda^2)  \rangle_P \,  \rd x }\,. 
\end{equation}
For the denominator, we have 
\begin{equation} \label{eq:mcdenom}
	\int_\R \llangle \Im M_0(x + \ii \alpha \lambda^2) \rrangle \,  \langle \Im M_0(x + \ii \alpha \lambda^2)  \rangle_P \,  \rd x  = \pi \llangle \Im M_0(E_0 + \ii \alpha \lambda^2) \rrangle + \mathcal{O}\big(\epsilon_0 + \Delta + \lambda^2/\Delta\big)
\end{equation}
from Lemma \ref{lem:rho0approx}. For the numerator, we use the assumption that $h(x) =  \llangle \Im M_0(x + \ii \alpha \lambda^2) A \rrangle$ has uniformly bounded Lipschitz constant for $x \in I_\Delta$ (recall \eqref{eq:hC1 norm}). Hence we find
\begin{equation} \label{eq:mcnumer}
\int_\R  h(x) \,  \langle \Im M_0(x + \ii \alpha \lambda^2) \rangle_P  \, \rd x = \pi h(E_0) + \mathcal{O}\big(\Delta + \lambda^2 /\Delta\big)
\end{equation}
completely analogously to \eqref{eq:rhointegral} and~\eqref{eq:rhocauchy}, using Assumption \ref{ass:stateandobs}~(i). 

Therefore, plugging \eqref{eq:mcdenom} and~\eqref{eq:mcnumer} into \eqref{eq:mcproofstart}, and using Assumption \ref{ass:H0}~(ii) together with Lemma \ref{lem:rho0approx}, we obtain
\begin{equation*}
\langle A \rangle_{\widetilde{P}_\lambda} = \langle A \rangle_{{P}^{\rm (mc)}_{\lambda}} + \mathcal{O}(\cE_{\rm mc}) \quad \text{with} \quad \cE_{\rm mc} := \epsilon_0 + \Delta + \lambda^2/\Delta\,. 
\end{equation*}
This concludes the proof of Theorem \ref{thm:PreT2}. \qed

\appendix
\section{Auxiliary results and additional proofs} \label{app:app}
\subsection{Auxiliary results}
In this subsection, we derive two technical lemmas, which are frequently used throughout the main text. 

The first one (Lemma \ref{lemma:M_bounds}) is concerned with properties of the self-consistent resolvent $M_\lambda(z)$ from the $\lambda$-dependent Quadratic Matrix Equation \eqref{eq:MDE} using Assumption \ref{ass:H0} on the unperturbed matrix $H_0$. {Recall that $\Delbound$ denotes the upper bound for the energy width $\Delta$ (cf.~Assumption~\ref{ass:stateandobs}).} 
\begin{lemma} \label{lemma:M_bounds}
	Let $z:= E+\ii \eta$ be a spectral parameter in $\C$ with $|z|\le C$, then the solution $M_{\lambda}(z)$ to \eqref{eq:MDE} satisfies the bounds
	\begin{equation} \label{eq:trM_bounds}
		\llangle |M_{\lambda}(z)|^2\rrangle \le \lambda^{-2}, \quad \llangle |M_{\lambda}(z)|\rrangle \le \lambda^{-1}, \quad \Im z >0.
	\end{equation}
	Moreover, assuming that $\Im z \ge 0$, $|\Re z - E_0| \le \Dela$, and $E_0$ lies in the admissible spectrum $\sigma_{\mathrm{adm}}^{(\kappa_0, c_0)}$ of $H_0$ the following estimates
	\begin{equation} \label{eq:trM_local}
		\Im m_\lambda(z)\gtrsim 1,\quad |m_\lambda(z)| \lesssim 1, \quad  |m_\lambda'(z)| \lesssim 1,
	\end{equation}
	with $m_\lambda(z) := \llangle M_{\lambda}(z) \rrangle $ hold for any fixed $0 < \lambda \le \lambda_*$, and all $N \ge N_\lambda \in \N$, uniformly in $z$. 
	
\end{lemma}
\begin{proof}[Proof of Lemma \ref{lemma:M_bounds}]
	First, we prove the \eqref{eq:trM_bounds} for $\Im z > 0$. Taking the imaginary part of the MDE \eqref{eq:MDE}, multiplying by $|M_\lambda|^2$ and taking the averaged trace yields
	\begin{equation}
		\Im m_\lambda(z) = \bigl(\Im z + \lambda^2\Im m_\lambda(z)\bigr)\llangle |M_{\lambda}(z)|^2\rrangle,
	\end{equation}
	which immediately implies the first bound in \eqref{eq:trM_bounds}. The second estimate in \eqref{eq:trM_bounds} follows from the first one by the Cauchy--Schwarz inequality. To extend \eqref{eq:trM_bounds} down to $\Im z = 0$, we first address the regularity\footnote{The comprehensive analysis of the MDE in \cite{AEK2019, AEK2020} shows that for any fixed $\lambda$, under the additional boundedness assumption $\lVert M_\lambda \rVert \lesssim 1$, the operator $M_\lambda(z)$ is $1/3$-H\"{o}lder continuous with a $\lambda$-dependent constant. However, for the purposes of proving the regularity of $m_\lambda$, the operator norm bound is not necessary.} of $m_\lambda(z) = \llangle M_\lambda(z)\rrangle$. 
	
	Differentiating the MDE \eqref{eq:MDE}, taking the trace, and evoking the first bound in \eqref{eq:trM_bounds}, we obtain
	\begin{equation} \label{eq:M'_bound}
		\bigl|m'_\lambda(z)\bigr| = \frac{\bigl|\llangle M_\lambda(z)^2 \rrangle\bigr|}{\bigl| 1 - \lambda^2 \llangle M_\lambda(z)^2 \rrangle\bigl|} 
		\le \frac{1}{2\lambda^4\llangle[\bigl] (\Im M_\lambda(z))^2 \rrangle[\bigr]} \lesssim \frac{1}{\lambda^4(\Im m_\lambda(z))^2} ,
	\end{equation}
	where we used \eqref{eq:trM_bounds} to deduce that $| 1 - \lambda^2 \llangle M_\lambda(z)^2 \rrangle| \ge 1 - \lambda^2 \Re\llangle M_\lambda(z)^2 \rrangle \ge 2\lambda^2\llangle (\Im M_\lambda(z))^2 \rrangle$, and the positive-definiteness of $\Im M_\lambda(z)$ to obtain the last inequality in \eqref{eq:M'_bound}.
	
	Together with the first estimate in \eqref{eq:trM_bounds}, \eqref{eq:M'_bound} implies that $\lambda^{4/3} \Im m_\lambda(z)$ is uniformly $1/3$-H\"{o}lder continuous in $\{z\in \C: |\Im z| \ge 0, |z| \le C\}$. This concludes the proof of \eqref{eq:trM_bounds} for $\Im z \ge 0$.

	Next, we turn to proving the first estimate in \eqref{eq:trM_local}. It follows from the MDE \eqref{eq:MDE} that 
	\begin{equation} \label{eq:MDE_traced}
		m_\lambda(z) = \llangle[\bigl] \bigl(H_0 - z - \lambda^2m_\lambda(z)\bigr)^{-1} \rrangle[\bigr],
	\end{equation}
	and $\Im z + \lambda^2\Im m_\lambda(z) \ge \Im z$.
	
	First, assume that $\Im z \ge \eta_0$ (recall \eqref{eq:rho0}). It follows from \eqref{eq:trM_bounds} that $\lambda^2|m_\lambda(z)| \le \lambda$, hence by suitably shrinking the threshold $\lambda_*$, we can assume that $|\Re [z+\lambda^2 m_\lambda(z)] - E_0| \le \Delb$ 
	for all $z$ satisfying $\Im z \ge \eta_0$, $|z|\le C$ and $|\Re z - E_0| \le \Dela$. Therefore
	\begin{equation} \label{eq:trM_est}
		m_\lambda(z) = m_0\bigl(z+ \lambda^2m_\lambda(z)\bigr) + \mathcal{O}(\epsilon_0) = m_0(z) + \mathcal{O}(\lambda + \epsilon_0),
	\end{equation}
	where the first step follows by \eqref{eq:rho0} and \eqref{eq:trM_bounds}, and in the second estimate we used that $E_0 \in \sigma_{\mathrm{adm}}^{(\kappa_0, c_0)}$, defined in \eqref{eq:admiss_spec}. In particular, taking the imaginary part of \eqref{eq:trM_est}, and using the positivity of $\rho_0$ in the admissible spectrum yields $\Im m_\lambda(z) \gtrsim 1 + \mathcal{O}(\lambda + \epsilon_0)$. Hence, from the $1/3$-H\"{o}lder continuity of $\lambda^{4/3}\Im m_\lambda(z)$, and \eqref{eq:trM_est} we deduce that
	\begin{equation} \label{eq:Imm_lambdalower}
		\lambda^2\Im m_\lambda(z) \gtrsim \lambda^2 + \mathcal{O}(\lambda^3 + \lambda^2\epsilon_0+\lambda^{2/3}\eta_0^{1/3}).
	\end{equation}
	for all $z$ with $|\Re z - E_0| \le \Dela$, $\Im z \ge 0$, and $|z|\le C$. Therefore, for a suitably small threshold $\lambda_*$ and any fixed $0<\lambda\le\lambda_*$, the first estimate in \eqref{eq:trM_local} is established for all $N$ satisfying $\eta_0(N)^{1/3} \lesssim \lambda^2$ with the implicit constant depending only on the constant in \eqref{eq:M'_bound}. Since $\eta_0(N)$ converges to zero, there exists $N_\lambda\in \N$ such that all $N\ge N_\lambda$ satisfy $\eta_0(N)^{1/3} \lesssim \lambda^2$. Furthermore, we obtain $\Im z + \lambda^2\Im m_\lambda(z) \gtrsim \lambda^2 \ge \eta_0$, hence the second estimate in \eqref{eq:trM_local} follows from \eqref{eq:rho0} and \eqref{eq:MDE_traced}.
	
	Finally, we prove the third estimate in \eqref{eq:trM_local}. Differentiating \eqref{eq:MDE_traced} with respect to $z$ yields
	\begin{equation} \label{eq:trMderiv}
		m'_\lambda(z) = \llangle[\bigl] \bigl(H_0 - z - \lambda^2m_\lambda(z)\bigr)^{-2} \rrangle[\bigr]\bigl(1+\lambda^2	m'_\lambda(z)\bigr),
	\end{equation}
	Im particular, the first factor on the right-hand side of \eqref{eq:trMderiv} is a normalized trace of $(H_0-\zeta)^{-2}$ with $\zeta := z + \lambda^2m_\lambda(z)$, satisfying $|\Re \zeta - E_0| \le \Delb$ and $\Im \zeta \ge 2c \lambda^2 \ge \eta_0$ for some positive constant $c\sim 1$ by the first and second estimates in \eqref{eq:trM_local}. Hence, integrating \eqref{eq:rho0}, we deduce that
	\begin{equation}
		\llangle (H_0-\zeta)^{-2} \rrangle = \int_ {\R+\ii c\lambda^2} \frac{ \Im\llangle  (H_0-w)^{-1} \rrangle}{\pi\,(w-\zeta)^2}\mathrm{d}w = \int_{\R+\ii c\lambda^2} \frac{ \Im m_0(w) }{\pi\,(w-\zeta)^2}\mathrm{d}w + \mathcal{O}\bigl(\lambda^{-2}\epsilon_0\bigr).
	\end{equation}
	Moreover, using the fact that $m_0(w)$ is the Stieltjes transform of the limiting density $\rho_0$, we conclude that
	\begin{equation} \label{eq:trMsquare}
		\llangle (H_0-\zeta)^{-2} \rrangle = \frac{1}{\pi}\int_{\R} \frac{ \rho_0(x) }{(x-\zeta)^2}\mathrm{d}x + \mathcal{O}\bigl(\lambda^{-2}\epsilon_0\bigr),
	\end{equation}
	In particular, since $\Re \zeta$ lies in the admissible spectrum of $H_0$, the integral on the right-hand side of \eqref{eq:trMsquare} admits the estimate
	\begin{equation} \label{eq:rhoderivint}
		\biggl|\frac{1}{\pi}\int_{\R} \frac{ \rho_0(x) }{(x-\zeta)^2}\mathrm{d}x\biggr| = \biggl|\int_{J}\frac{\rho_0'(x) -\rho_0'(\Re\zeta)}{x-\zeta} \mathrm{d}x\biggr| + \biggl|\int_{J} \frac{\rho_0'(\Re\zeta)\,\Im\zeta}{|x-\zeta|^2}\mathrm{d}x \biggr| + \mathcal{O}(\kappa_0^{-2}),
	\end{equation}
	where $J:= [\Re\zeta-c\kappa_0, \Re\zeta+c\kappa_0] \subset I_{\kappa_0}$. Here we used that the kernel $(x-\Re\zeta)/|x-\zeta|^2$ is odd around $\Re\zeta$ and the integrability of $\rho_0$ to estimate the integral over $\R\backslash J$, while the boundary term resulting from integration by parts over $J$ is bounded by $\kappa_0^{-1}\mathcal{O}(\lVert \rho_0\rVert_{C^1(J)}) \lesssim \kappa_0^{-1}$. Observe that the second term on the right-hand side of \eqref{eq:rhoderivint} is bounded by $\mathcal{O}(|\rho_0'(\Re\zeta)|)$, while the first term is bounded by $\mathcal{O}(L)$, where $L$ is the Lipschitz constant of $\rho_0'$ on the interval $[E_0-\kappa_0, E_0+\kappa_0]$. Therefore, $|\llangle (H_0-\zeta)^{-2} \rrangle | = \mathcal{O}(1)$.
	Finally, rearranging the identity \eqref{eq:trMderiv} now yields
	\begin{equation}
		m'_\lambda(z) = \frac{\llangle[\bigl] \bigl(H_0 - z - \lambda^2m_\lambda(z)\bigr)^{-2} \rrangle[\bigr]}{1 - \lambda^2\llangle[\bigl] \bigl(H_0 - z - \lambda^2m_\lambda(z)\bigr)^{-2} \rrangle[\bigr]} = \mathcal{O}(1).
	\end{equation}
	This concludes the proof of Lemma \ref{lemma:M_bounds}.
\end{proof}
 
 We conclude this section by evaluating the denominator in \eqref{eq:kernelstate}. 
 
 \begin{lemma} \label{lem:rho0approx}
Under Assumptions \ref{ass:H0} and \ref{ass:stateandobs}~(i) (recalling the notation $\alpha = \pi \rho_0(E_0)$) it holds that 
\begin{equation} \label{eq:rhoapprox}
	\begin{split}
\left| \int_{\R} \llangle \Im M_0(x + \ii \alpha \lambda^2) \rrangle \langle \Im M_0(x +\ii \alpha \lambda^2) \rangle_P \,  \rd x -   \pi^2 \rho_0(E_0) \right| =  \mathcal{O}\big(\epsilon_0 + \Delta + \lambda^2/\Delta\big) \,, \\
\left| \pi^2 \rho_0(E_0) -   \pi \llangle \Im M_0(E_0 + \ii \alpha \lambda^2) \rrangle\right| = \mathcal{O}\big(\epsilon_0 + \Delta + \lambda^2/\Delta\big) \,. 
	\end{split}
\end{equation}
 \end{lemma}
\begin{proof}
We only prove the first relation in \eqref{eq:rhoapprox}. The argument for the second estimate is analogous and hence omitted.

We have that
\begin{equation} \label{eq:rhointegral}
\begin{split}
&\int_{\R} \llangle \Im M_0(x + \ii \alpha \lambda^2) \rrangle \langle \Im M_0(x +\ii \alpha \lambda^2) \rangle_P \,  \rd x \\
= &\int_{\R}  \Im m_0(x + \ii \alpha \lambda^2)  \langle \Im M_0(x +\ii \alpha \lambda^2) \rangle_P \,  \rd x + \mathcal{O}(\epsilon_0) \\
= &  \sum_{\mu_j \in I_\Delta} \langle \bm u_j ,P \bm u_j \rangle \int_{\R}  \rd y \rho_0(y)\int_{\R} \rd x  \frac{\alpha \lambda^2}{(x - y)^2 + (\alpha \lambda^2)^2} \frac{\alpha \lambda^2}{(x - \mu_j)^2 + (\alpha \lambda^2)^2} + \mathcal{O}(\epsilon_0) \\
= & \pi^2 \sum_{\mu_j \in I_\Delta} \langle \bm u_j ,P \bm u_j \rangle \int_{\R}  \rd y \rho_0(y) \frac{1}{\pi}\frac{2\alpha \lambda^2}{(y - \mu_j)^2 + (2\alpha \lambda^2)^2} + \mathcal{O}(\epsilon_0) \,,
\end{split}
\end{equation}
where we used $\int_\R \langle \Im M_0(x +\ii \alpha \lambda^2) \rangle_P \,  \rd x = \pi$ and \eqref{eq:rho0} to go to the second line. To go to the third line, we employed spectral decomposition \eqref{eq:specdecH0} of $H_0$ and used Assumption \ref{ass:stateandobs}~(i). Next, using that $\mu_j \in I_\Delta$ and regularity of $\rho_0$ within $I_{2 \Delta}$, we can evaluate the integral in the last line of \eqref{eq:rhointegral} as
\begin{equation} \label{eq:rhocauchy}
\int_{\R}  \rd y \rho_0(y) \frac{1}{\pi}\frac{2\alpha \lambda^2}{(y - \mu_j)^2 + (2\alpha \lambda^2)^2} = \rho_0(E_0) + \mathcal{O}\big(\Delta + \lambda^2/\Delta\big)\,,
\end{equation}
by adding and subtracting $\rho_0(E_0)$ and estimating  $I_{2\Delta}$ and  $I_{2 \Delta}^c$ separately. 

Combining \eqref{eq:rhointegral} with \eqref{eq:rhocauchy} and using that $\Tr[P] = 1$ concludes the argument. 
\end{proof}

\subsection{Proof of Proposition \ref{prop:LL}} \label{app:LL}Our proof closely follows \cite[Appendix~B]{multiG}, \cite[Section~5.2]{iid}, and \cite[Section~6.2]{equipart},\footnote{In fact, for $z_1, z_2$ in the \emph{bulk} of the self consistent density of states $\rho(x) := \pi^{-1} \lim_{\eta \to 0^+} \llangle \Im M(x + \ii \eta) \rrangle$, Proposition \ref{prop:LL} has already been proven for so called \emph{regular} matrices $A$ in \cite[Proposition 4.4]{equipart}. Here, we provide the proof uniformly in the spectrum and for arbitrary matrices.} hence, we only give the main steps. Note that we are only interested in a \emph{global law}, i.e.~the spectral parameters $z_1, z_2$ have imaginary parts uniformly bounded away from zero. In particular, we can simply afford the norm bounds $\Vert G_i \Vert \lesssim 1$ and $\Vert M_i \Vert \lesssim 1$. 

As a preparation for our argument, we recall the definition of the \emph{second order renormalization}, denoted by \emph{underline}, from \cite{ETHpaper}. For functions $f(W), g(W)$ of the random matrix $W$, we define 
\begin{equation} \label{eq:underline}
	\underline{f(W) W g(W)} := f(W) W g(W) - \widetilde{\E} \big[  (\partial_{\widetilde{W}}f)(W)\widetilde{W} g(W) + f(W)  \widetilde{W} (\partial_{\widetilde{W}}g)(W) \big]\,,
\end{equation}
where $\partial_{\widetilde{W}}$ denotes the directional derivative in the direction of the GUE matrix $\widetilde{W}$ that is independent of $W$. The expectation is taken w.r.t.~the matrix $\widetilde{W}$. Note that if $W$ itself is a GUE matrix, then $\E \underline{f(W)Wg(W)} = 0$, while for $W$ with a general distribution, this expectation is independent of the first two moments of $W$. In other words, the underline renormalizes the product $f(W)W g(W)$ to the second order.  We remark that underline \eqref{eq:underline} is a well-defined notation if the `middle' $W$ to which the renormalization refers is unambiguous. This is the case in our proof, since the functions $f, g$ are resolvents, i.e.~not involving explicitly monomials in $W$.

Moreover, we note that $\widetilde{\E} \widetilde{W} R \widetilde{W} = \llangle R \rrangle$ and furthermore, that the directional derivative of the resolvent is given by $\partial_{\widetilde{W}} G = -G \widetilde{W} G$.
For example, in the special case $f(W) = 1$ and $g(W) = (W + D -z)^{-1} = G$, we thus have
\begin{equation*}
	\underline{WG} = WG + \llangle G \rrangle G
\end{equation*}
by definition of the underline in \eqref{eq:underline}. 
Using this underline notation in combination with the identity $G (W + D - z) = I $ and the defining equation \eqref{eq:MDEnolambda} for $M$, we have
\begin{equation} \label{eq:start}
	G = M - M \underline{W G} + M \llangle G-M \rrangle G  = M - \underline{GW}M + G \llangle G-M \rrangle M\,. 
\end{equation}

Moreover, we have the following lemma, the proof of which is given at the end of this section. 
\begin{lemma}[Representation as full underlined, cf.~Lemma 5.6 in \cite{iid}] \label{lem:underline}Under the notations and assumptions of Proposition \ref{prop:LL}, we have that
\begin{equation} \label{eq:underlinerep}
\left( G_1BG_2-\Big(M_1BM_2+\frac{M_1M_2 \llangle M_1BM_2 \rrangle }{1- \llangle M_1M_2 \rrangle}\Big)\right)_{\bm x \bm y} = - \big( \underline{G_1 B' M_2 W G_2}\big)_{\bm x \bm y} + \mathcal{O}_\prec \left(\mathcal{E}_{\rm iso}\right)
\end{equation}
with $\mathcal{E}_{\rm iso} := 1/\sqrt{N}$ for some bounded deterministic matrix $B' \in \C^{N \times N}$.  
\end{lemma}

Having this approximate representation of the lhs.~of \eqref{eq:LL} as a full underlined term at hand, we turn to the proof of \eqref{eq:LL} via a (minimalistic) cumulant expansion; see \cite[Eq.~(4.32)]{multiG} and \cite[Eq.~(5.24)]{iid}. 

	Let $p \in \N$ be arbitrary. Then, abbreviating the lhs.~of \eqref{eq:underlinerep} by $\mathcal{Q}_{\bm x \bm y}$, we obtain
	\begin{equation} \label{eq:minexp} 
		\E \big|\mathcal{Q}_{\bm x \bm y} \big|^{2p} 
		\lesssim  \, \E \,  \widetilde{\Xi} \,  \big\vert \mathcal{Q}_{\bm x \bm y} \big\vert^{2p-2} + \sum_{|\boldsymbol{l}| + \sum(J \cup J_*) \ge 2} \E \, \Xi(\boldsymbol{l}, J, J_*) \big\vert  \mathcal{Q}_{\bm x \bm y} \big\vert^{2p-1 - | J \cup J_*|} + \mathcal{O}_\prec\big(\mathcal{E}_{\rm iso}^{-2p}\big) \,,
	\end{equation}
	where the summation in  \eqref{eq:minexp} is taken over tuples $\boldsymbol{l} \in \Z^2_{\ge 0}$ and multisets of tuples $J, J_* \subset \Z^2_{\ge 0} \setminus \{(0,0)\}$, for which we set $\partial^{(l_1,l_2)} := \partial_{ab}^{l_1} \partial_{ba}^{l_2}$, $|(l_1, l_2)| = l_1 + l_2$, $\sum J = \sum_{\boldsymbol{j} \in J} |\boldsymbol{j}|$. 
	Moreover, we denoted
	\begin{align}  \label{eq:tildeXi}
		\widetilde{\Xi} := &\frac{\left\vert \big(  G_1 {B}' G_1 B G_2\big)_{\boldsymbol{x}\boldsymbol{y}} \big( G_1  G_2 \big)_{\boldsymbol{x}\boldsymbol{y}} \right\vert + \left\vert\big( G_1 {B}'  G_2 \big)_{\boldsymbol{x}\boldsymbol{y}} \big(G_1 B G_2  G_2  \big)_{\boldsymbol{x}\boldsymbol{y}}\right\vert}{N}  \\[2mm]
		&+ \frac{ \left\vert \big( G_1 B' G_2^*B^* G_1^*\big)_{\boldsymbol{x}\boldsymbol{x}} \big(G_2^*  G_2\big)_{\boldsymbol{y}\boldsymbol{y}} \right\vert + \left|  \big(G_1 B'  G_1^*\big)_{\boldsymbol{x}\boldsymbol{x}} \big( G_2^*A^* G_1^*  G_2\big)_{\boldsymbol{y}\boldsymbol{y}}\right|}{N}\,, \nonumber
	\end{align}
and defined $\Xi(\boldsymbol{l}, J, J_*)$ via
	\begin{align} \label{eq:Xi}
		\Xi:=N^{-(|\boldsymbol{l}| + \sum(J \cup J_*) + 1)/2} \sum_{ab} &\big|\partial^{\boldsymbol{l}} \big[(G_1 {B}')_{\boldsymbol{x}a}\big(G_2\big)_{b \boldsymbol{y}} \big]\big|  \\
		&\times \prod_{\boldsymbol{j} \in J} \big|\partial^{\boldsymbol{j}}  \big(G_1 B G_2\big)_{\boldsymbol{x}\boldsymbol{y}}\big|   \prod_{\boldsymbol{j} \in J_*} \big|\partial^{\boldsymbol{j}} \big(G_2^* B^* G_2^*\big)_{\boldsymbol{y}\boldsymbol{x}}  \big| \,.  \nonumber
	\end{align}

Now, by a simple norm bound, $\Vert G \Vert \lesssim 1$, we find that 
\begin{equation} \label{eq:Xitildebound}
\widetilde{\Xi} \prec \mathcal{E}_{\rm iso}^2\,. 
\end{equation}

For $\Xi$, our goal is to show that 
\begin{equation} \label{eq:Xibound} 
\Xi(\bm l, J, J_*) \prec \mathcal{E}_{\rm iso}^{1 + | J \cup J_*|} \,. 
\end{equation}
First, we have the naive bounds
\begin{equation} \label{eq:naive}
\big|\partial^{\boldsymbol{l}} \big[(G_1 {B}')_{\boldsymbol{x}a}\big(G_2\big)_{b \boldsymbol{y}} \big]\big|  + \big|\partial^{\boldsymbol{j}}  \big(G_1 B G_2\big)_{\boldsymbol{x}\boldsymbol{y}}\big| +  \big|\partial^{\boldsymbol{j}} \big(G_2^* B^* G_2^*\big)_{\boldsymbol{y}\boldsymbol{x}}  \big| \prec 1
\end{equation}
and hence
\begin{equation*}
\Xi \prec N^{-(|\boldsymbol{l}| + \sum(J \cup J_*) +1)/2} N^2 = N^{(4 - |\boldsymbol{l}| )/2}N^{-(1+ \sum(J \cup J_*))/2} = N^{(4 - |\boldsymbol{l}| )/2} \mathcal{E}_{\rm iso}^{1+ \sum(J \cup J_*)} \,. 
\end{equation*}
Thus, for $|\bm l| \ge 4$, we find the naive bounds \eqref{eq:naive} to be sufficient for \eqref{eq:Xibound},  since trivially $| J \cup J_*| \le \sum (J \cup J_*)$. For $|\bm l| \le 3$, we perform the summation $\sum_{ab}$ a bit more carefully, e.g., recalling the norm bound $\Vert G \Vert \lesssim 1$, as 
\begin{equation*}
\sum_{a} |G_{\bm x a}| \le N^{1/2} \sqrt{\sum_{a} |G_{\bm x a}|^2} = N^{1/2} (|G|^2)_{\bm x \bm x} \lesssim N^{1/2}
\end{equation*}
instead of the naive $\sum_{a} |G_{\bm x a}| \lesssim N$. Indeed, using the condition $|\bm l| + \sum (J \cup J_*) \ge 2$, we can check all the cases $|\bm l| \le 3$ explicitly and find 
\begin{equation*}
\sum_{ab} \big|\partial^{\boldsymbol{l}} \big[(G_1 {B}')_{\boldsymbol{x}a}\big(G_2\big)_{b \boldsymbol{y}} \big]\big|  \prod_{\boldsymbol{j} \in J} \big|\partial^{\boldsymbol{j}}  \big(G_1 B G_2\big)_{\boldsymbol{x}\boldsymbol{y}}\big|   \prod_{\boldsymbol{j} \in J_*} \big|\partial^{\boldsymbol{j}} \big(G_2^* B^* G_2^*\big)_{\boldsymbol{y}\boldsymbol{x}}  \big| \prec N^{2} \, N^{-(4 - |\boldsymbol{l}| )/2}\,,
\end{equation*}
from which we conclude \eqref{eq:Xibound}. 

Plugging \eqref{eq:Xibound} together with \eqref{eq:Xitildebound} into \eqref{eq:minexp}, using a Young inequality and recalling that $p$ was arbitrary, we deduce that 
\begin{equation*}
|\mathcal{Q}_{\bm x \bm y}| \prec \mathcal{E}_{\rm iso} = \frac{1}{\sqrt{N}}\,, 
\end{equation*}
i.e.~we have proven Proposition \ref{prop:LL}. \qed

 It remains to give the proof of Lemma \ref{lem:underline}. 
\begin{proof}[Proof of Lemma \ref{lem:underline}]
Applying \eqref{eq:start} to $G_2$, we thus find that
\begin{equation*}
	G_1 \widetilde{B} G_2 = G_1 \widetilde{B} M_2 - G_1 \widetilde{B} M_2 \underline{W G_2} + G_1 \widetilde{B} M_2 \llangle G_2-M_2\rrangle G_2
\end{equation*}
for  $\widetilde{B} = \mathcal{X}_{12}[B]$, where we introduced the linear operator
\begin{equation} \label{eq:X12def}
	\mathcal{X}_{12}[C] := \big(1 - \llangle M_1 \, \cdot \, M_2\rrangle \big)^{-1}[C]\quad \text{for} \quad C \in \C^{2N \times 2N}\,.
\end{equation}
Extending the underline to the whole product, we obtain
\begin{align*}
	G_1 \widetilde{B} G_2 = &M_1 \widetilde{B} M_2 + (G_1 - M_1) \widetilde{B} M_2 - \underline{G_1 \widetilde{B} M_2 W G_2} \\
	&+ G_1 \widetilde{B} M_2 \llangle G_2 - M_2 \rrangle G_2 + G_1 \llangle G_1 \widetilde{B} M_2\rrangle G_2\,,
\end{align*}
from which we conclude 
\begin{align} \label{eq:globallawfinal}
	G_1 BG_2 = &M_1 \mathcal{X}_{12}[B] M_2 + (G_1 - M_1) \mathcal{X}_{12}[B] M_2 - \underline{G_1 \mathcal{X}_{12}[B]M_2 W G_2} \\
	&+ G_1 \mathcal{X}_{12}[B] M_2 \llangle G_2 - M_2 \rrangle G_2 + G_1 \llangle (G_1-M_1) \mathcal{X}_{12}[B] M_2 \rrangle  G_2 \,. \nonumber
\end{align}
To proceed, we note that 
$$
B' := \mathcal{X}_{12}[B] = B+\frac{ \llangle M_1BM_2 \rrangle }{1- \llangle M_1M_2 \rrangle}
$$ 
has bounded norm, $\Vert B' \Vert \lesssim1 $, since for $|\Im z_1|, |\Im z_2| \gtrsim 1$, $\mathcal{X}_{12}$ is a bounded operator (see \cite[Lemma B.5]{iid} and \cite[Appendix~A.2]{equipart}). Then, using the norm bounds $\Vert G_i \Vert \lesssim 1$ and $\Vert M_i \Vert \lesssim 1$ together with the \emph{single resolvent global law} \eqref{eq:singlegloballaw} (see also \cite[Theorem~2.1]{slowcorr}) for the second, fourth and fifth term in \eqref{eq:globallawfinal}, we conclude the desired. 
\end{proof}

\renewcommand*{\bibname}{References}

\bibliographystyle{amsplain}
\bibliography{RefPreT}

\end{document}